\documentclass[11pt]{article}
\usepackage[left=1in, right=1in, top=1in, bottom=1in, margin=1in]{geometry}

\usepackage{algorithm}
\usepackage{tcolorbox}
\usepackage{amsmath,amssymb,xspace,graphicx,relsize,bm,xcolor,breqn,algpseudocode,multirow}

\usepackage{tikz}
\usetikzlibrary{quantikz2}
\usepackage{adjustbox}

\usepackage{tcolorbox}
\newcommand{\Exp}{\mathop{\mathbb{E}}}

\usepackage[margin=1in]{geometry}

\usepackage{amsthm}
\usepackage{dsfont}
\usepackage{array}
\usepackage{makecell}

\newcommand{\Sh}{\ensuremath{\mathsf{Stab}}}

\newcommand{\poly}{\ensuremath{\mathsf{poly}}}

\newcommand{\R}{\ensuremath{\mathbb{R}}}

\newcommand{\id}{\ensuremath{\mathbb{I}}}

\usepackage{enumerate}

\usepackage[pagebackref]{hyperref}
\usepackage[usestackEOL]{stackengine}
\usepackage{thm-restate,mathrsfs}
\usepackage{enumerate}
\usepackage{array}
\usepackage{parskip}
\def\01{\mathbb{F}_2}
\newcommand{\ketbra}[2]{|#1\rangle\langle#2|}
\hypersetup{
	colorlinks,
	linkcolor={blue!100!black},
	citecolor={red!100!black},
}

\newcommand{\be}{\begin{equation}}
\newcommand{\ee}{\end{equation}}
\newcommand{\ba}{\begin{array}}
\newcommand{\ea}{\end{array}}
\newcommand{\bea}{\begin{eqnarray}}
\newcommand{\eea}{\end{eqnarray}}

\usepackage{mathtools}

\DeclareMathOperator{\Tr}{Tr}
\newcommand{\ra}{\rangle}
\newcommand{\la}{\langle}

\newcommand{\norm}[1]{\left\lVert#1\right\rVert}

\newcommand{\calF}{{\cal F }}

\newcommand{\calS}{{\cal S }}

\newcommand{\calP}{{\cal P }}

\newcommand{\FF}{\mathbb{F}}

\newcommand{\ZZ}{\mathbb{Z}}
\newcommand{\CC}{\mathbb{C}}

\newcommand{\gowers}[2]{\textsc{Gowers}\left({#1},{#2}\right)}
\newcommand{\AC}{\textsf{AC}}
\newcommand{\nac}{\textsf{nac}}

\declaretheorem[numberwithin=section]{theorem}
\declaretheorem[sibling=theorem]{lemma}
\declaretheorem[sibling=theorem]{claim}
\declaretheorem[sibling=theorem]{proposition}
\declaretheorem[sibling=theorem]{corollary}
\declaretheorem[sibling=theorem]{conjecture}

\theoremstyle{definition}

\declaretheorem[sibling=theorem]{fact}
\declaretheorem[sibling=theorem]{remark}

\usepackage{thm-restate} %To repeat theorems, lemmas etc....

\makeatletter
\def\widebreve{\mathpalette\wide@breve}
\def\wide@breve#1#2{\sbox\z@{$#1#2$}%
     \mathop{\vbox{\m@th\ialign{##\crcr
\kern0.08em\brevefill#1{0.8\wd\z@}\crcr\noalign{\nointerlineskip}%
                    $\hss#1#2\hss$\crcr}}}\nolimits}
\def\brevefill#1#2{$\m@th\sbox\tw@{$#1($}%
  \hss\resizebox{#2}{\wd\tw@}{\rotatebox[origin=c]{90}{\upshape(}}\hss$}
\makeatletter
\title{Polynomial-time tolerant testing stabilizer states}

\begin{document}

\author{
Srinivasan Arunachalam\\[2mm]
IBM Quantum\\
\small Almaden Research Center\\
\small \texttt{Srinivasan.Arunachalam@ibm.com}
\and
Arkopal Dutt\\[2mm]
IBM Quantum\\
\small   Cambridge, Massachusetts\\
\small \texttt{arkopal@ibm.com}
}

% \date{\today}

\maketitle
\begin{abstract}
We consider the following task: suppose an algorithm is given copies of an unknown $n$-qubit quantum state $\ket{\psi}$ promised $(i)$ $\ket{\psi}$ is $\varepsilon_1$-close to a stabilizer state in fidelity or $(ii)$ $\ket{\psi}$ is  $\varepsilon_2$-far from all stabilizer states, decide which is the case. We show that for every $\varepsilon_1>0$ and $\varepsilon_2\leq \varepsilon_1^C$, there is a $\poly(1/\varepsilon_1)$-sample and $n\cdot \poly(1/\varepsilon_1)$-time algorithm that decides which is the case (where $C>1$ is a universal constant). Our proof includes a new definition of Gowers norm for quantum states, an inverse theorem for the Gowers-$3$ norm of quantum states and new bounds on stabilizer covering for structured subsets of Paulis using results in additive~combinatorics. 

\end{abstract}

{\small \tableofcontents}
\section{Introduction}
\paragraph{Classical property testing.} The field of property testing was initiated by the seminal work, Blum, Luby and Rubinfield~\cite{blum1990self} who studied the following task: decide if a Boolean function $f:\01^n\rightarrow \01$ satisfies a global property given only query-access to $f$ or is far from all functions satisfying this property. More formally, a property can be associated with a collection of functions $\mathcal{P}\subseteq \{f:\01^n\rightarrow \01\}$ and the goal is: given query access to a function $g$, decide if $(i)$ $g$ is in the property, i.e., $g\in \calP$ or $(ii)$ $g$ is $\varepsilon$-far from the property under some distance metric $d$, i.e., $\max_{f\in \calP}d(f,g)\geq \varepsilon$, promised one of them is the case. Since its initiation, the field of property testing has  been influential in theoretical computer science. %such as, communication complexity, learning theory, coding theory and computational complexity.
In particular,~\cite{blum1990self} showed that \emph{$3$ queries} suffices to test the property of linear functions and this was the basis for the PCP~theorem.

Subsequently, Parnas, Ron and Rubinfield~\cite{parnas2006tolerant} introduced the notion of \emph{tolerant} property testing wherein they modified the usual property testing goal as follows: decide if the unknown function $g$ is $\varepsilon_1$-close to the property (i.e., $\max_{f\in \calP}d(f,g)\leq \varepsilon_1$) or $\varepsilon_2$ far from the property. Note that the standard notion of property testing is when one fixed $\varepsilon_1=1$. Since its introduction, several works have investigated the strengths and limitations of tolerant property testing, showing that standard property testers can sometimes be made tolerant as well as giving hardness results in the tolerant framework~\cite{alon2000efficient,diakonikolas2007testing,fischer2004testing,fischer2005tolerant,goldreich2003three}. 

\paragraph{Quantum property testing.} This generalizes the classical setting, wherein the goal is to test a property of a \emph{quantum state} instead of a Boolean function. Here, a property $\calP$ is a collection of pure quantum states and the goal in \emph{intolerant} testing is: given copies of an unknown $\ket{\psi}$, decide if $(i)$ $\ket{\psi}$ is in the property, i.e., $\ket{\psi}\in \calP$ or $(ii)$ $\ket{\psi}$ is $\varepsilon$-far from the property, i.e., $\max_{\ket{\phi}\in \calP} |\langle \psi|\phi\rangle|^2\leq 1-\varepsilon$, promised of them is the case. Although~(intolerant) quantum property testing is more nascent than classical property testing, influential results have been produced which have had applications in other areas of quantum computing. 
It is well known that stabilizer states~\cite{gross2021schur}, product states~\cite{harrow2013testing} and matrix product states~\cite{sw2022testing} are efficiently testable. Other efficiently (intolerant) testable states include low-degree phase states~\cite{arunachalam2022optimal}, high-temperature Gibbs states of local Hamiltonians~\cite{anshu2021sample,haah2022optimal}, and states produced by Clifford circuits with few $T$ gates~\cite{grewal2023efficient}, for which efficient learning algorithms are known, implying testing~algorithms.

Since near-term quantum devices are noisy, a natural implementation of a quantum algorithm will inherently have some noise, motivating the study of \emph{tolerance} in quantum algorithms. Akin to classical property testing, we initiate \emph{tolerant quantum property testing} in this work. As far as we know, only a handful of works~\cite{grewal2023improved,gilyen2019distributional} have touched upon this topic. In fact, we do not know of any natural class of states that are testable in the tolerant framework. Since stabilizer states are efficiently learnable~\cite{montanaro2017learning}, simulatable and (intolerant) testable, it motivates the  question 
\begin{quote}
\begin{center}
\emph{Is there an efficient (i.e., polynomial-time) tolerant tester for stabilizer states?}
\end{center}
\end{quote} 
If one does not require efficiency, then one could use  shadow tomography~\cite{aaronson:shadow,buadescu2021improved} that uses $O(n^2)$ copies and exponential time. However, a holy grail in property testing is \emph{constant-copy testers} (i.e., the query/sample complexity is independent of $n$). Similar to~\cite{blum1990self}, constant-copy tolerant testers could potentially be useful in understanding the quantum PCP conjecture. To this end,~\cite{gross2021schur} showed that stabilizer states are \emph{intolerant} testable  using $6$ copies, however these algorithms are not robust. We know conditional hardness of learning states which are \emph{noisy/close} to stabilizers, but not exactly stabilizers~\cite{gollakota2022hardness,chen2024stabilizer,hinsche2022single}, motivating the ambitious~question 
\begin{quote}
\begin{center}
\emph{Can we tolerant test stabilizer states with {constantly many} copies?}
\end{center}
\end{quote}

\subsection{Main result}
Our main result is a positive answer to the above questions. We give a constant-sample and time-efficient algorithm that tolerantly tests stabilizer states. We state our theorems below before comparing with prior work. Below, we will let $\Sh$ denote the class of $n$-qubit stabilizer~states.
\begin{restatable}{theorem}{tolerantstabtesting}\label{thm:testphase2}
There is a constant $C>1$ for which the following is true. For any $\varepsilon_1 > 0$ and $\varepsilon_2\leq \varepsilon_1^C$, there exists an algorithm that~given
$\poly(1/\varepsilon_1)$ copies of an $n$-qubit quantum state $\ket{\psi}$, decides if $\max_{\ket{\phi}\in \Sh}|\langle\phi |\psi\rangle|^2 \geq  \varepsilon_1$ or $\max_{\ket{\phi}\in \Sh}|\langle\phi |\psi\rangle|^2 \leq  \varepsilon_2$  using $O(n\cdot \poly(1/\varepsilon_1))$~gates.
\end{restatable}
We remark that this theorem covers the regime when the gap between $\varepsilon_1,\varepsilon_2$ is multiplicative, in which case our sample complexity is independent of $n$ and depends only on the gap. 
 The regimes which are not covered by this theorem is when the additive gap between $\varepsilon_1,\varepsilon_2$  is an arbitrary polynomial in $n$, which we leave for tackling as part of future work. Additionally, we did not make efforts in optimizing the constants in this proof. Prior to our work, there were two results in this direction: Gross, Nezami and Walter~\cite{gross2021schur} showed a result assuming $\varepsilon_1=1$; a follow-up work of Grewal, Liang, Iyer, Kretschmer~\cite{grewal2023improved} showed one can improve the analysis of~\cite{gross2021schur} to obtain an improved property tester but their analysis could only handle $\varepsilon_1^6\geq 1/4$. Both these works fall short of the usual definition of \emph{tolerant} property testing (where $\varepsilon_1,\varepsilon_2$ need \emph{not} be constants and also $\varepsilon_1$ need not at least be a universal constant) and our theorem addresses this gap. As far as we are aware, this is one of the first results on tolerant quantum property testing with sample complexity independent~of~$n$.

\begin{table}[h]
\small
\centering
 \begin{tabular}{|c | c  | c  |c|c|} 
 \hline
  & \makecell{$\varepsilon_1$} & $\varepsilon_2$ &Sample &Time \\ [0.5ex] 
 \hline
 \makecell{~\cite{gross2021schur}} & $1$  & $< 1-16\sqrt{1-\varepsilon_1}$ &$O(1)$&$O(n)$   \\ 
 \hline
 \makecell{~\cite{grewal2023improved}}   & $\geq (1/4)^{1/6}$ & $\leq (4\varepsilon_1^6-1)/3$ &$O(1)$&$O(n)$\\ \hline
 \makecell{\emph{This work}} & $\geq \tau$ & $\leq \tau^C$&$\poly(1/\tau)$&$n\cdot \poly(1/\tau)$ \\ 
 \hline
 \makecell{Concurent work~\cite{chen2024stabilizer}} & $\geq \tau$ & $\leq \tau-\varepsilon$&$n\cdot (1/\tau)^{\log (1/\tau)}$& $n^3/\varepsilon^2\cdot (1/\tau)^{\log (1/\tau)}$ \\ 
 \hline
 \end{tabular}
 \caption{Parameters for tolerant testing $n$-qubit stabilizer states. The goal of the tolerant tester is to decide if $(i)$ $\max_{\ket{\phi}\in \Sh}|\langle\phi |\psi\rangle|^2 \geq  \varepsilon_1$ or $(ii)$ $\max_{\ket{\phi}\in \Sh}|\langle\phi |\psi\rangle|^2 \leq  \varepsilon_2$ . The $\tau$ in the final two rows could be an arbitrary function of $n$ and $C>1$ is a universal constant.} 
 \label{tab:summary_results_paper}
\end{table}
Concurrently, Chen, Gong, Ye and Zhang~\cite{chen2024stabilizer} put out a paper regarding agnostic learning stabilizer states (a task much harder than tolerant testing, but could be used as a way for tolerant testing). We make a couple of remarks comparing our works:  $(i)$ the proof techniques in both papers are vastly different, $(ii)$ their paper can handle ranges of $\varepsilon_1,\varepsilon_2$ where the difference is inverse polynomial in $n$, but ours requires a \emph{multiplicative} guarantee that $\varepsilon_2\leq \varepsilon_1^C$ for some constant $C>1$, $(iii)$  their sample  complexity grows with $n$, whereas our sample complexity is a constant, and $(iv)$ furthermore when $\tau=1/\poly(n)$, our algorithm has time complexity that is polynomial in $n$ whereas their complexity is quasipolynomial in $n$. 

Finally, they also have a \emph{conditional} hardness result showing that one cannot hope for a polynomial-time algorithm for $\tau=1/\poly(n)$ by assuming that a non-standard variant of learning parity with noise is hard.  Interestingly, we show that in the weaker task of \emph{testing}, one can obtain a \emph{polynomial-time}~tester to decide if $\ket{\psi}$ is $\tau$-close to a stabilizer state in fidelity or $\tau^C$-far away.

\subsection{A warm up result}
Before we discuss our result, we first discuss classical low-degree testing, then give a ``warm-up" result (which is new as far as we are aware) by assuming that the unknown state $\ket{\psi}$ is a quantum phase state. We subsequently highlight a few issues which we had to overcome for arbitrary quantum states, motivating our final~algorithm.

\subsubsection{Classical low-degree testing}
A seminal result  of Samorodnitsky~\cite{samorodnitsky2007low} connects two well-known notions in mathematics and  theoretical computer science: \emph{Gowers norm} and \emph{low-degree testing}. In the task of low-degree testing, the goal is to determine if a Boolean function $f: \mathbb{F}_2^n \rightarrow \mathbb{F}_2$ has at most degree $d$ (in its $\FF_2$ representation) or far from all such functions. A prior work of Alon, Kaufman, Krivelevich, Litsyn and Ron~\cite{alon2003testing} gave a low-degree property testing algorithm that implicitly used the \emph{definition} of Gowers norm without explicitly mentioning it, but the work of Samorodnitsky~\cite{samorodnitsky2007low} made this connection explicit. For $k\geq 1$, the Gowers-$k$ norm (denoted $U^{k}$) of a Boolean function $f:\01^n\rightarrow \{-1,1\}$ is defined as:  
\begin{equation}
    \|{f}\|^{2^{k}}_{U^{k}} = \mathop{\Exp}_{x,h_1,\ldots,h_k \in \FF_2^n} \Big[\prod_{S\subseteq [k]}f\big(x+\sum_{i\in S}h_i\big)\Big].
\end{equation}
Intuitively, the $U^k$ norm can be viewed as taking an expectation of $k$-fold discrete derivatives of the Boolean function $f$ over its inputs. Gowers  norm has found several applications in theoretical computer science and mathematics~\cite{bogdanov2010pseudorandom,sam2006gowers,tulsiani2014quadratic,chen2009gowers,trevisan2009guest,green2007distribution}. 
Based on the intuition above, it is not hard to see that if $\|f\|_{U^3}$ is small, then $f$ is far from a degree-$2$ polynomial, but what about a converse statement? The main technical contribution~\cite{samorodnitsky2007low,gowers2024marton} 
was to show: if a function $f$ has high Gowers norm, i.e., $\|f\|_{U^3}\geq \varepsilon$, then $f$ correlates well with a degree-$2$ polynomial, i.e., $\Exp_x[f(x) (-1)^{g(x)}]\geq \poly(\varepsilon)$, where $g: \FF_2^n \rightarrow \FF_2$ is a degree-$2$ polynomial.
This statement was eventually used for \emph{tolerant} property testing degree-$2$ Boolean functions. The starting point of this work was in considering if a quantum-analogue of this converse statement~holds. 

\subsubsection{Testing phase states}
Instead of giving a proof sketch of the result of Samorodnitsky~\cite{samorodnitsky2007low}, we first use ideas inspired by his paper and solve a simpler tolerant property testing problem: what if the unknown state $\ket{\psi}$ was promised to be a \emph{phase state}, i.e., 
$$
\ket{\psi_f}=\frac{1}{\sqrt{2^n}}\sum_{x\in \01^n}f(x)\ket{x},
$$
for some unknown $f:\01^n\rightarrow \{-1,1\}$. Then, could we test if $\ket{\psi_f}$ is $\varepsilon_1$-close to a stabilizer state in fidelity or $\varepsilon_2$-far away? As far as we are aware, sample efficient bounds on this task were unknown prior to our work. We now sketch a high-level idea of how to tolerantly test these states.\footnote{The possibility of  such a property tester assuming $\ket{\psi}$ is a phase state, was implicitly mentioned in~\cite{mehraban2024quadratic}.}  Inspired by~\cite{samorodnitsky2007low}, we give a proof sketch below for this task.
\begin{enumerate}
\item \textbf{Samples versus queries.} One interesting observation we start with here is that: classically one needs \emph{query} access to $f$ in order to compute its Gowers norm (intuitively one needs to make \emph{correlated} queries of the form $f(x),f(y),f(x+y)$, hence needing queries), however in order to obtain an $\varepsilon$-approximation of the $U^3$ norm, it suffices for a quantum algorithm to just obtain $O(1/\varepsilon^2)$ \emph{copies} of the unknown phase state $\ket{\psi_f}$ and it need not have query access to $f$. So with just copies of $\ket{\psi_f}$, a tester can decide if $\|f\|_{U^3}$ is large or small. Below our goal is to show that, if $\|f\|_{U^3}$ is large then $f$ is close to a degree-$2$ function (which in particular corresponds to a stabilizer state).
\item \textbf{Structural observations.} The main observation here is that using Fourier analysis on the Boolean cube, one can show that~\cite[Corollary~6.6]{samorodnitsky2007low} if $\|f\|_{U^3}\geq \varepsilon$, then
\begin{align}
    \label{eq:implicationofclassicalFU3}
    \Exp_{x,y}\Big[\sum_{\alpha,\beta}\widehat{f}_x(\alpha)^2 \widehat{f}_y(\beta)^2 \widehat{f}_{x+y} (\alpha+\beta)^2\Big]\geq \poly(\varepsilon),
\end{align}
i.e., the derivative function $f_x$ (in the direction $x$) is approximately \emph{linear}.\footnote{Note that when $f$ is degree-$2$, then the derivative functions are \emph{exactly} linear and~\cite{samorodnitsky2007low} shows that large Gowers norm implies approximately linearity.} 
\item \textbf{Choice function.} At this point, we use the same crucial idea as in the proof of~\cite[Lemma~6.7]{samorodnitsky2007low}: probabilistically define a choice function $\phi:\01^n\rightarrow \01$ as: $\Pr[\phi(x)=\alpha]=\widehat{f}_x(\alpha)^2$ and using Eq.~\eqref{eq:implicationofclassicalFU3} one can show the \emph{existence} of a choice function $\phi$ which is approximately linear, i.e., a function $\phi$ for which
$$
L(\phi)=\Pr_{x,y}[\phi(x)+\phi(y)=\phi(x+y):\widehat{f}_x(\phi(x))^2, \widehat{f}_y(\phi(y))^2, \widehat{f}_{x+y}(\phi(x+y))^2 \geq \varepsilon/4]\geq \poly(\varepsilon).
$$
\item \textbf{Additive combinatorics.}  At this point, similar to the proof of~\cite{samorodnitsky2007low}, we use a well-known result of Balog-Szemeredi-Gowers (BSG)~\cite{balog1994statistical} as well as the polynomial Freiman–Ruzsa conjecture~\cite{ruzsa1999analog} (which we discuss further below) to show the existence of a symmetric  matrix $B \in \mathbb{F}_2^{n\times n}$ for which $\phi(x)=B\cdot x$. Using ideas from linear algebra~\cite{viola2011combinatorics,hatami2019higher}, one can conclude that $|\langle \psi_g |\psi_f \rangle|^2\geq \poly(\varepsilon)$ for some quadratic function $g:\01^n\rightarrow \01$.\footnote{We remark that the original result by \cite{samorodnitsky2007low} had an exponential dependence of $2^{-\poly(1/\varepsilon)}$ on $\varepsilon$ which was improved to $\poly(\varepsilon)$ by the recent breakthrough work of Gowers, Green, Manners and Tao~\cite{gowers2023conjecture} that resolved the polynomial Freiman–Ruzsa conjecture.} 
\end{enumerate}
Hence, the above arguments show that if $\|f\|_{U^3}\geq \varepsilon$, then $f$ is $\poly(\varepsilon)$-correlated with a degree-$2$ phase state (which in particular is a stabilizer state). This gives an $\poly(1/\varepsilon)$ tolerant property testing algorithm \emph{assuming} that the input state is of the form $\frac{1}{\sqrt{2^n}}\sum_x f(x)\ket{x}$. 
Now we describe the challenges in proving a similar statement for an arbitrary state $\ket{\psi}$. 

\subsubsection{Challenges in going beyond phase states.} Consider an  $n$-qubit pure state $\ket{\psi}$ (that is not necessarily a phase state). The proof sketch above does not generalize for the following~reasons:
\begin{enumerate}[$(i)$]
    \item A notion of Gowers norm is unclear for quantum states; although one could define a natural analogue similar to the classical setting, properties of this definition were unknown (and its relation to to stabilizer testing unclear) prior to our work,
    \item The notion of a random \emph{choice function} above has been crucial in several works in this area~\cite{samorodnitsky2007low,tulsiani2014quadratic,viola2011combinatorics}, but for arbitrary quantum states, these functions do not even make sense. In particular, this choice function above was motivated by the following fact: if an algorithm was given copies of a phase state upon which it performed Bell sampling, then the samples are all of the form $\{(z,\phi(z)):z\in \01^n\}$ (which coincides with the stabilizers of a degree-$2$ phase state) for an arbitrary $\phi:\01^n\rightarrow \01$. In particular for a degree-$2$ phase state $\ket{\psi}$, for every $z$ there always \emph{exists} a $\phi(z)$ such that $W_{z,\phi(z)}$ (up to a phase) stabilizes $\ket{\psi}=\sum_x f(x)\ket{x}$, i.e., $|\langle \psi|W_{z,\phi(z)}|\psi\rangle|=1$, where $W_{a,b}$ is the $n$-qubit Weyl operator given by $X^a Z^b$ (up to a global phase). But this is clearly false for more general states. In fact, even if one considers stabilizer states, then it is not true.  Consider the simplest  GHZ states
    $$
    \frac{1}{\sqrt{2}}\big(\ket{0^n}+\ket{1^n}\big),
    $$
    then the stabilizers of this state are $\{X^aZ^b: a\in \{0^n,1^n\}, \langle b,1^n\rangle=0\}$. So the first $n$ bits of $(a,b)$ are either $0^n$ or $1^n$, so for all $x\neq \{0^n,1^n\}$, $\phi(x)=0$, and defining such a choice \emph{function} would not work. 
    \item In item $(4)$ of the proof sketch above for testing phase states, we need to \emph{indirectly} show that the set of Paulis corresponding to $\{(x,\phi(x))\}$ can be covered efficiently by an union of stabilizer groups, also called \emph{stabilizer covering}. Using a recent result of King, Gosset, Kothari, Babbush~\cite{king2024triply} regarding stabilizer coverings of Paulis corresponding to arbitrary states, we would have obtained a testing algorithm which could have handled a gap of $\varepsilon_1\geq 1/k$ versus $\varepsilon_2\leq n^{-k}$. However, our goal is to handle $\varepsilon_1,\varepsilon_2$ that is explicitly independent of $n$ (even though $k$ may be a function of $n$). In fact we remark that several works in literature have considered the task of producing a stabilizer covering of subsets of Paulis corresponding to quantum states~\cite{chen2024optimal,king2024triply} or Pauli channels~\cite{flammia2020efficient}, but are only able to produce generic  upper bounds by the uncertainty principle~\cite{king2024triply} or show that it is equivalent to other notions of hardness of finding sets of commuting Paulis~\cite{chen2024optimal}, all of which would introduce a dimension-$n$ dependence in the tolerance parameters of the property testing  algorithm.
\end{enumerate}

In the next section, we now show how to address the challenges that we mentioned above.

\subsection{Gowers norms, inverse theorems and tolerant testing}
\paragraph{Gowers norm.} On a conceptual level, our first contribution is a definition of Gowers norm: for an arbitrary quantum state $\ket{\psi}=\sum_x f(x)\ket{x}$ where $f=(f(x))_x$ is an $\ell_2$-normed vector, one can define its Gowers norm as follows  
\begin{equation}
    \gowers{\ket{\psi}}{k} = 2^{n/2} \left[ \Exp_{x,h_1,h_2,\ldots,h_k \in \mathbb{F}_2^n} \prod_{\omega \in \mathbb{F}_2^k} C^{|\omega|} f(x + \omega \cdot h) \right]^{1/2^{k}},
\end{equation}
where $C^{|\omega|} f = f$ if $|\omega| := \sum_{j \in [k]} \omega_k$ is even and is $\overline{f}$ if $|\omega|$ is odd with $\overline{f}$ denoting the complex conjugate of $f$. Using this definition we make a few observations: $(i)$ quantum states $\ket{\psi}$ produced by unitaries containing gates up to the $d$-th level of the diagonal Clifford hierarchy satisfy  $\gowers{\ket{\psi}}{d+1}^{2^{d+1}}=1$,  and $(ii)$ more importantly for our main result, we observe that $\gowers{\ket{\psi}}{3}$ is very closely related to the output of Bell sampling. For a quantum state $\ket{\psi}$, the characteristic distribution is defined as $p_\psi(x)=2^{-n}\cdot |\langle \psi|W_x|\psi\rangle|^2$, and the Weyl distribution is defined as its convolution $q_\psi=4^n (p_\psi\star p_\psi)$. It is well-known that Bell-difference sampling~\cite{gross2021schur} allows us to estimate the expression $\Exp_{x\sim q_\psi}[\langle \psi|W_x|\psi\rangle^2]$. Our main observation is that, 
\begin{align}
\label{eq:gowerswithbellsampling}
\gowers{\ket{\psi}}{3}^{16}\leq \Exp_{x\sim q_\psi}[\langle \psi|W_x|\psi\rangle^2]\leq \gowers{\ket{\psi}}{3}^{8},
\end{align}
in particular at least for $k=3$, the Gowers-$3$ norm serves as a good-proxy for the output of Bell sampling. Although, the Gowers-$3$ norm is essentially a ``proxy" for the quantity $\Exp_{x\sim q_\psi}[\langle \psi|W_x|\psi\rangle^2]$, which was the basis of all previous works on Bell sampling~\cite{grewal2022low,grewal2023improved}, we believe that the introduction of this is new for quantum computing and might motivate further connections to other areas (like we mention as part of the open questions).\footnote{We remark that all the results in this paper were motivated by Gowers norm and its utility for testing Boolean functions before we realized a close connection between Gowers-$3$ norm of states, and outputs of well known subroutines such as Bell sampling and Bell difference sampling. The perspective of Gowers norm is thus mainly helpful as a ``ladder" to obtain our results for \emph{stabilizer state testing}.}

The reason why prior work analyzed the expression $\Exp_{x\sim q_\psi}[\langle \psi|W_x|\psi\rangle^2]$ is because it was shown in~\cite{gross2021schur,grewal2023improved} that this expression serves as a \emph{lower bound} for stabilizer fidelity,~i.e.,
\begin{align}
\label{eq:GIKLlowerboundstabfidelity}
\max_{\ket{\phi}\in \Sh}|\langle \phi|\psi\rangle|^2 \geq  \big(4\cdot \Exp_{x\sim q_\psi}[\langle \psi|W_x|\psi\rangle^2]-1\big)/3\geq  \big(4\cdot \gowers{\ket{\psi}}{3}^{16}-1\big)/3,
\end{align}
but this lower bound only makes sense as long as $\gowers{\ket{\psi}}{3}^{16}\geq 1/4$, and this is why these results only gave tolerant testing algorithms for $\varepsilon_1^6\geq 1/4$. In other words, prior works implied that stabilizer fidelity of $\ket{\psi}$ being $\geq \gamma$ implied that $\gowers{\ket{\psi}}{3}^8 \geq \gamma^6$. However, what if Gowers-3 norm of $\psi$ is $\gowers{\ket{\psi}}{3}^8 \geq \gamma$? Does it imply a lower bound on stabilizer fidelity?  Our second main contribution is a positive answer to this question and improving the lower bound in Eq.~\eqref{eq:GIKLlowerboundstabfidelity}; so one need not assume $\varepsilon_1\geq 1/4$, and we call this the \emph{inverse Gowers-$3$ theorem} for quantum states.

\paragraph{Inverse Gowers theorem.} Like how, large Gowers $U^3$ norm for functions, implies high correlation with degree-$2$ Boolean functions, our inverse Gowers-$3$ theorem for quantum states shows~that 
$$
\max_{\ket{\phi}\in \Sh}|\langle \phi|\psi\rangle|^2 \geq \poly(\gowers{\ket{\psi}}{3}),
$$
i.e., high $\gowers{\ket{\psi}}{3}$ implies that the unknown $\ket{\psi}$ is close to a stabilizer state. In order to prove this, we put together a few structural results that we discuss below. Below, we will assume that  $\gowers{\ket{\psi}}{3}^{8}\geq \gamma$. Our overall proof strategy and the different sets obtained obtained as a consequence of high Gowers-$3$ norm is summarized in Figure~\ref{fig:gowers_sets}.

\begin{figure}
    \centering
    \includegraphics[width=0.95\linewidth]{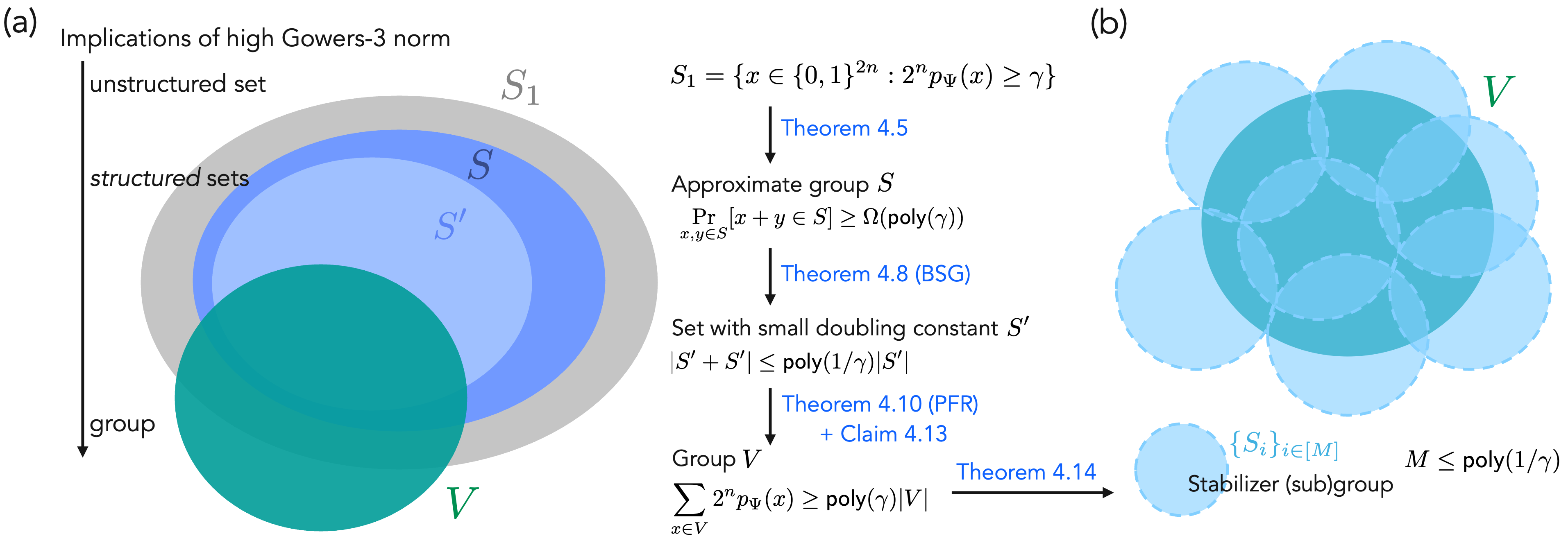}
    \caption{Illustration of sets, their properties and theorem statements obtained en route to show an inverse theorem of the Gowers-$3$ norm of quantum states. (a) Sets obtained starting from $S_1=\{x \in \mathbb{F}_2^{2n}:2^n p_\Psi(x) \geq \gamma \}$ as a consequence of high Gowers-$3$ norm of an $n$-qubit quantum state $\ket{\psi}$ i.e., $\gowers{\ket{\psi}}{3}^8 \geq \gamma$, are depicted. Each of the sets $S_1$, $S$, $S'$, and group $V$ shown are large with sizes being greater than $\poly(\gamma) \cdot 2^n$. Proof strategy of the inverse theorem is indicated by the order of theorems and their implications. (b) Illustration of a stabilizer covering of the group $V$ obtained in (a) is shown along with the theorem commenting on the size of the stabilizer covering.}
     \label{fig:gowers_sets}
\end{figure}
\subsubsection{Large Gowers norm to approximate subgroup}
\label{sec:largegowersstart}
Our starting point is the subset
\begin{align}
\label{eq:defnofS1}
 S_1 = \{ (x,\alpha) : x, \alpha \in \mathbb{F}_2^{n} \, \text{ and } |\langle \psi|W_{x,\alpha}|\psi\rangle|^2\geq \gamma/4 \}.
\end{align}
It is not hard to see that $|S_1|\geq \gamma/2\cdot 2^n$.  Recall that in~\cite{samorodnitsky2007low}, for every $x$, there exists an $\alpha$ such that $(x,\alpha)\in S_1$, so $\alpha$ was the associated label for the choice function. In our setting, we probabilistically define a \emph{choice set}  $S$ as follows: for every $(x,\alpha)\in S_1$, include $(x,\alpha)\in S$ with some carefully chosen probability. We first show that with high probability $(i)$ $|S|$ is large, i.e., $|S| \in [\gamma^2,1/\gamma]\cdot 2^n$; $(ii)$ next we  analyze what is the probability that such an $S$ is closed under addition on $\FF_2$, i.e., for a uniformly random choice of $(x,\alpha)$ and $(y,\beta)$ in $S$, what is the probability of $(x+y,\alpha+\beta)\in S$? To this end, for every  $Z\subseteq \01^{2n}$ we first define 
$$
L(Z)=\Pr_{(x,\alpha),(y,\beta)\in S}[x+y,\alpha+\beta\in S].
$$ 
With some technical work, we are able to show that $$
\gowers{\ket{\psi}}{3}^{16}\leq \poly(1/\gamma)\cdot \Exp_{S}[L(S)]+O(\gamma^2).
$$
Since the Gowers norm was assumed to be large, the LHS of the quantity above is at least $\gamma^2$, we have that $\Exp_{S}[L(S)]\geq \poly(\gamma)$. Furthermore, we are also able to show the existence of $S \subseteq S_1$ for which $(i)$ $|S| \in [\gamma^2,1/\gamma]\cdot  2^n$,  $(ii)$ $S$ consists of $(x,\alpha)$s for which $|\langle \psi|W_{x,\alpha}|\psi\rangle|^2\geq \gamma/4$, $(iii)$ $S$ contains $\id$ and $(iv)$  $S$ is \emph{approximately} a subgroup, i.e.,
\begin{align}
    \label{eq:approximatelyadditive}
\Pr_{(x,\alpha), (y,\beta) \in S}[ (x+y,\alpha+\beta) \in S]\geq \poly(\gamma).
\end{align}

\subsubsection{Approximate subgroup to small doubling sets}
In order to go from the subset $S$ that is \emph{approximately} closed under addition to an Abelian subgroup (and eventually a stabilizer subgroup), we carry out the following three step process.
\begin{enumerate}
\item We first use the seminal BSG~\cite{balog1994statistical} theorem that states that if a set $S$ satisfies Eq.~\eqref{eq:approximatelyadditive}, then there exists $S' \subseteq S$ of large size that is nearly closed under addition (in other words,  if many $x,y\in S$ also satisfy that $x+y\in S$, then there is a large subset $S'\subseteq S$ with not much larger $|S'+S'|$). In particular, there exists $S'$ of size $|S'| \geq \gamma\cdot  |S|$ with $|S' + S'| \leq \poly(1/\gamma) \cdot |S|$.
\item  So far, we have only shown that $S'$  is a subset of $S$ and satisfies $|S'+S'| \leq \poly(1/\gamma) |S'|$. To  show that the set $S'$ identified in step $(1)$ can be covered by a \emph{group}, we then use the recent breakthrough result of Gowers, Green, Manners and Tao~\cite{gowers2023conjecture}, which resolved the polynomial Freiman–Ruzsa conjecture and showed that $S'$ can be covered by at most $\poly(1/\gamma)$ many cosets of some subgroup $V \subset \FF_2^{2n}$ with size $|V| \leq |S'|$. This allows us to conclude that there exists a coset of $V\subset \FF_2^{2n}$, which we will denote by $z + V$ such that
\begin{equation}
    \left| S' \cap (z+V) \right| \geq \eta |S'| \geq \poly(\gamma) 2^n,
\end{equation}
where $\eta = c \cdot \poly(1/\gamma)$ for some constant $c \geq 1$.
\item We first remark that although the result of~\cite{gowers2024marton} might already seem sufficient for our purposes, we are not done yet. Note that an \emph{affine} shift of a subgroup need not be a stabilizer subgroup and neither is it clear that it can be covered by stabilizer group (potentially one needs $\poly(n)$ many stabilizer groups to cover). As far as we are aware, the affine shift seems \emph{necessary} for results with the flavour of~\cite{gowers2023conjecture}. Furthermore, an application of uncertainty principle (like used in the recent works~\cite{king2024triply,chen2024optimal}) would imply a covering of $n^{1/\gamma^2}$. Our main observation is that, one can get rid of the affine shift for our purposes in the following sense. We observe that
$$
\Exp_{y \in V}\left[ 2^n p_\Psi(y+z) \right]\geq \frac{|S' \cap (z+V)|}{|S'|}  \cdot {\gamma}/{4} \geq  \eta\cdot \gamma/4,
$$
and then using calculations involving the symplectic Fourier transform of $p_\Psi$, we show that 
$$
\max_{z\in \01^{2n}}\Exp_{y \in V}\left[ 2^n p_\Psi(y+z) \right]= \max_{z\in V}\Exp_{y \in V}\left[ 2^n p_\Psi(y+z) \right].
$$
This implies that the probability mass over $V$ is large i.e., $\Exp_{y \in V}\left[ 2^n p_\Psi(y) \right] \geq \eta \cdot \gamma/4 \geq \poly(\gamma)$.
\end{enumerate}

\subsubsection{Approximate group covering to stabilizer covering}\label{sec:intro_stab_covering}
Given our previous observation, we now shift gears from trying to produce a stabilizer covering of the Paulis corresponding to $S'$, and instead attempt to cover the Paulis of the subgroup $V$. In order to comment on the size of the stabilizer covering required, one needs to understand the maximal size of a set containing pairwise anti-commuting Paulis in the subgroup $V$, denoted by $\nac(V)$. Particularly, the stabilizer covering of $V$ would then be at most $2^{\nac(V)}$. We first note that there exists a Clifford unitary $U$ and a pair of integers $(k,m)$ (with $k+m \leq n$) such that $V$ can be equivalently expressed as $UVU^\dagger = \calP^{k} \times \la Z_{k+1}, \ldots, Z_{k+m}\ra$ with $2k + 1 = \nac(V)$, where $\calP^k$ is the set of all $k$-qubit Paulis, and $Z_i$ is the $n$-qubit Pauli with $Z$ acting on the $i$th qubit and trivially elsewhere. Our main contribution is to then obtain an upper bound on $\sum_{W_x \in UVU^\dagger} 2^n p_{\tilde{\Psi}}(x)$ where $\tilde{\Psi} = U \Psi U^\dagger$, by analyzing the contribution from expectations of the $k$-qubit reduced density matrix of $\tilde{\Psi}$ on the subsystem over the first $k$ qubits against Paulis in $UVU^\dagger$. This allows us to show that
$$
\poly(\gamma) 2^{2k + m} = \poly(\gamma) |V| \leq \sum_{x \in V} 2^n p_\Psi(x) = \sum_{W_x \in \calP^k \times \la Z_{k+1},\ldots,Z_{k+m}\ra} 2^n p_{\tilde{\Psi}}(x) \leq 2^{k+m},
$$
where we have indicated the size of $V$ being $|V| = 2^{2k+m}$ on the left-hand size. Given our bound on $ \sum_{x \in V} 2^n p_\Psi(x)$, we immediately observe that  $k = O(\log(1/\gamma))$ and thereby giving us that the stabilizer covering of $V$ is at most $4^{k}=\poly(1/\gamma)$. Let us denote the stabilizer covering as $\{S_{j}\}_{j \in [4^{k}]}$. We now choose $S_i$ which maximizes $\sum_{y \in V \cap S_i}\left[ 2^n p_\Psi(y) \right]$. For this choice of $S_i$, we get $\sum_{y \in V \cap S_i} p_\Psi(y) \geq \poly(\gamma) \cdot \Omega(1/\poly(1/\gamma)) = \poly(\gamma)$. Using a well-known analysis, one can show that: since $\sum_{y \in S_i} p_\psi(y)$ is large for a Lagrangian subspace $S_i$, the stabilizer fidelity of $\ket{\psi}$ is at least $\poly(\gamma)$. This concludes the proof of our inverse Gowers theorem for arbitrary quantum states.

\subsubsection{Tolerant testing algorithm}
As we mentioned in our warm-up discussion, one can compute the Gowers-$3$ norm of an arbitrary quantum state $\ket{\psi}$ via Bell sampling on $\ket{\psi} \otimes \ket{\psi^\star}$. We can also compute $\Exp_{x\sim q_\psi}[\langle \psi|W_x|\psi\rangle^2]$ via Bell difference sampling using only copies of $\ket{\psi}$. In the case where $\max_{\phi\in \Sh}|\langle\phi |\psi\rangle|^2 \leq  \varepsilon_2$,  our inverse Gowers theorem implies that $\gowers{\ket{\phi}}{3}^8\leq \poly(\varepsilon_2)$ and in the  case where  $\max_{\phi\in \Sh}|\langle\phi |\psi\rangle|^2 \geq  \varepsilon_1$, recent works~\cite{grewal2023improved,gross2021schur} imply that $\gowers{\ket{\phi}}{3}^8\geq  \poly(\varepsilon_1)$. By our choice of $\varepsilon_2\leq {\poly(\varepsilon_1)}$, we ensure that the promise gap is $\poly(\varepsilon_1)$. We can also instead estimate $\Exp_{x\sim q_\psi}[\langle \psi|W_x|\psi\rangle^2]$, which produces a similar inverse theorem. So the testing algorithm takes $\delta=\poly(1/\varepsilon_1)$ many copies of $\ket{\psi}$ to estimate $\Exp_{x\sim q_\psi}[\langle \psi|W_x|\psi\rangle^2]$  well-enough. By an appropriate choice of constants, our main theorem follows and the overall gate-complexity involves $n$ many Pauli and CNOT gates.

\subsubsection{Quantum-inspired conjecture in additive combinatorics} 
Another contribution of our work is, we propose a conjecture in additive combinatorics, motivated by the result and proof strategy of our inverse theorem for the Gowers-$3$ norm of quantum states. In our proof of the inverse theorem, after identifying the approximate subgroup $S$, we used results from combinatorics 
to obtain the set with small doubling $S'$ and then the subgroup $V$. Then, to comment on the stabilizer covering of $V$, we analyzed the expectation of different Weyl operators against the state $\ket{\psi}$. But, what if we had attempted to proceed in a more \emph{combinatorial} fashion? This is the gap that the conjecture attempts to fill in.

To discuss this conjecture, we first need the following notation. For an arbitrary subset $A\subseteq \01^{2n}$ and integer $t\geq 1$,  define $tA=\{a_1+a_2+\cdots +a_t:\{a_i\}_{i\in [t]}\in A\}$). Then the conjecture is
\begin{conjecture}
\label{conjecture:doubling}
Let $A \subseteq \01^{2n}$ be s.t.~$|A|\geq 2^n/K$ and $|2A|\leq K |A|$. Then, $\nac(2A)\leq \poly(K,\nac(A))$.
\end{conjecture} 

We asked experts in additive combinatorics about this conjecture who said it was open and an interesting \emph{structural question} about small doubling sets.\footnote{We asked authors of~\cite{gowers2023conjecture} who suggested this conjecture was stronger than polynomial Frieman-Ruzsa theorem and  closer in spirit to the polynomial Bogulybov conjecture (which remains an  open question).} Using this conjecture, we show below how to $(i)$ obtain a stabilizer covering of $V$ of size $2^{\poly(1/\gamma)}$, and $(ii)$ show that $\ket{\psi}$ has a stabilizer fidelity of at least $2^{-\poly(1/\gamma)}$ (note that the result in Section~\ref{sec:intro_stab_covering} already showed $\ket{\psi}$ has a stabilizer fidelity of at least $\poly(\gamma)$, so this is exponentially weaker). Given that we have an \emph{alternate} technique to show a better stabilizer covering of $V$, it then gives evidence that the conjecture is true. So we feel it is worthwhile discussing this conjecture. 
 We now sketch how we arrived at this conjecture and why it implies a stabilizer covering (with a detailed proof in Appendix~\ref{app_sec:conjecture}). 

As in Section~\ref{sec:intro_stab_covering}, we comment on the stabilizer covering of $V$ by looking at the number of anti-commuting Paulis in $V$. Our main contribution is obtaining better bounds on the size of anti-commuting set of Paulis in structured subsets. We show that, if  an arbitrary subset $A\subseteq \01^{2n}$ can be covered by $M$ cosets (or translates) of another set $B \subseteq \mathbb{F}_2^{2n}$, then one can upper bound $\nac(A)$~as
$$
\nac(A)\leq 2M\cdot \nac(B).
$$ 
Courtesy the result of \cite{gowers2023conjecture} and the Ruzsa's covering lemma~\cite{taovu2006comb}, we can show that $V$ itself can be covered by $\poly(1/\gamma)$ translates of $2S'$. This then allows us to show that 
$$
\nac(V)\leq \poly(1/\gamma)\cdot \nac(2S').
$$
\emph{Assuming Conjecture~\ref{conjecture:doubling}}, we have that $\nac(2S')\leq \poly(1/\gamma,\nac(S'))$. Additionally, we have that $\nac(S')=O(1/\gamma)$ using an uncertainty principle. This then yields $\nac(V)=\poly(1/\gamma)$, and allows us to show that one can cover the Paulis corresponding to $V$ by a disjoint union of $M = 2^{\poly(1/\gamma)}$ many stabilizer groups $\{S_a\}_{a \in [M]}$. We now choose $S_i$ which maximizes $\sum_{y \in V \cap S_i}\left[ 2^n p_\Psi(y) \right]$. For this choice of $S_i$, we get $\sum_{y \in V \cap S_i} p_\Psi(y) \geq \poly(\gamma) \cdot 2^{-\poly(1/\gamma)}$. One can then show that: since $\sum_{y \in S_i} p_\psi(y)$ is large for a Lagrangian subspace $S_i$, the stabilizer fidelity of $\ket{\psi}$ is at least $2^{-\poly(1/\gamma)}$.

\subsection{Discussion and open questions}
\textbf{Outlook.} In this work, we settled the question of tolerant testing stabilizer states. As mentioned in the introduction, this is the first step towards discovering which quantum states are testable in a tolerant manner. We remark we are only aware of few works~\cite{asadi2024quantum} that have combined the rich field of additive combinatorics and quantum algorithms. To this end, we highlight two things.\\[1mm]
$(i)$ \emph{Choice sets.} As far as we are aware, the notion of choice sets has not appeared before in literature. Often it is required~\cite{sam2006gowers,samorodnitsky2007low,tulsiani2014quadratic} that there is a ``product distribution" between the first and second argument of the set $S_1$ (defined in Section~\ref{sec:largegowersstart}). Particularly, Samorodnitsky's proof requires that the marginal distribution on the first $n$ bits is uniform and independent of the second $n$ bits. Additionally, for every $x$, there is \emph{at least} one $\alpha$ which would be sampled. Our work shows that even with a joint distribution over the two arguments, there exists an approximate subgroup if Gowers-$3$ norm is high. 
We suspect this could be potentially useful in testing distribution classes, linearity over unknown subspaces or even higher-order inverse Gowers~theorems.\\[1mm] $(ii)$ \emph{Quantum-inspired conjecture} in combinatorics might be of independent interest to the fields of theoretical computer science and additive combinatorics.

Both of these could inspire further research in the intersection of these two~areas.

\textbf{Further questions.} Our work opens up several interesting directions for future work.
\begin{enumerate}
\item \textbf{Handling additive gaps.} Our tolerant testing algorithm requires  $\varepsilon_2\leq \poly(\varepsilon_1)$, but what if $\varepsilon_1-\varepsilon_2$ was inverse polynomial in $n$? Proving a tolerant tester in this regime is open.
%\item \textbf{Improving the  $\varepsilon_2$.} A natural open question is, could we remove the assumption that $\varepsilon_2 \leq 2^{-1/\varepsilon_1}$? Proving a better inverse Gowers theorem and better bounds on anticommuting Paulis in subsets of Paulis could potentially improve our bounds.

\item \textbf{Constants in our protocol.}  We did not make any attempt to optimize the constants in our testing protocol, it would be good to make these constants as small as possible.

\item \textbf{Pseudorandomness.} Classically, there is an intricate connection between a function that has small Gowers norm and it being pseudorandom. It is an interesting direction if one could show quantum states with low Gowers norms are pseudorandom quantum states. 

\item \textbf{Going beyond stabilizer states.} The starting point of our work was~\cite{samorodnitsky2007low} which showed that if a function has high $\|f\|_{U^3}$ norm then there exists a quadratic Boolean function $q$ such that $(-1)^{q(x)}$  has high correlation with $f$. In our setting, high Gowers-$3$ norm of a quantum state $\ket{\psi}$ implied $\ket{\psi}$ has high stabilizer fidelity. Subsequently, Kim, Li and Tidor~\cite{kim2023cubic} showed an inverse Gowers-$4$ theorem, i.e., if $f$ has high $\|f\|_{U^4}$ then there is a cubic phase polynomial $(-1)^{c(x)}$ (where $c$ is a degree-$3$ Boolean function) which has high correlation with~$f$. In the quantum setting, an interesting direction is to show that if a state has high Gowers-$4$ norm then it has high fidelity with a state which is the output of an IQP~circuit.
\item \textbf{Computing larger Gowers norms.} We showed how to compute the $U^3$ norm of a state given access to copies of it. What is the sample complexity for computing the Gowers-$U^k$ norm for $k\geq 4$? This could potentially be useful for testing low-degree phase states.
\item \textbf{Better shadow tomography?} Like we mentioned above, recent works of~\cite{king2024triply,chen2024optimal} gave algorithms for shadow tomography on $n$ qubits for arbitrary subsets of observables (with running time scaling as $2^n$). We suspect that our  inverse Gowers theorem, its extension to mixed states in Theorem~\ref{thm:large_set_gamma_mixed_states} and better covering bounds could improve their time complexity to $\poly(n)$, albeit in the ``heuristic" model of learning. We leave this open.\footnote{Here, we are referring to the recently-introduced model by Nanashima~\cite{nanashima2021theory} wherein the goal is for the algorithm to be correct on a $(1-\delta)$-fraction of the inputs, allowing the algorithm to have run time dependent on~$\poly(1/\delta)$.}
\end{enumerate}

\paragraph{Note added.} The first version of our paper showed how results in additive combinatorics could be used to obtain a tolerant tester for stabilizer states, but could only handle exponential gaps and required a conjecture in additive combinatorics. Subsequently, along with Sergey Bravyi we showed~\cite{arunachalam2024note}  how to circumvent this conjecture and gave a  tolerant tester that could handle polynomial gaps. This was concurrent with~\cite{helsen2024} that showed a similar tolerant tester directly building on~\cite[Version~2]{ad2024tolerant} by studying the Lovasz-theta number of the graph corresponding to the subgroup $V$ and their graph products. Soon after our note~\cite{arunachalam2024note}, there was another work~\cite{mehraban2024improved} which also showed  a tolerant tester using the toolset from~\cite[Version~1]{ad2024tolerant} but with a different approach.

\paragraph{Acknowledgements.}  Recently, Jop Briet and Jonas Helsen reached out with a proof sketch of~how they might improve~\cite[Version~1]{ad2024tolerant} by studying the Lovasz-theta number of the graph corresponding to the subgroup $V$ and its graph products. Upon this discussion,  we realized that we could achieve a tolerant tester using a note by Sergey Bravyi from much earlier. We are very grateful for this email exchange which motivated us to improve~\cite[Version~1]{ad2024tolerant} to obtain this~version. 

S.A. and A.D. thank the Institute for Pure and Applied Mathematics (IPAM) for its hospitality throughout the long program “Mathematical and Computational Challenges in Quantum Computing” in Fall 2023, during which part of this work was initiated. This work was done (in part) while S.A.~was visiting the Simons Institute for the Theory of Computing, supported by DOE QSA grant \#FP00010905. We thank Sergey Bravyi, Ewout Van den Berg for several discussions on anti-commuting subsets of Paulis, Terence Tao, Jyun-Jie Liao, Kaave Hosseini for clarifications on the recent results~\cite{gowers2023conjecture,gowers2024marton,liao2024improved}, Isaac L. Chuang, Sabee Grewal, and Theodore (Ted) J. Yoder for useful discussions.

\section{Preliminaries}
\subsection{Notation and lemmas} 
Let $[n]$ denote the set $\{1,\ldots,n\}$. For $a\leq b$ and $k\geq 1$, we denote $[a,b]\cdot k$ to be the set $\{ak,\ldots,bk\}$. We denote the finite field with the elements $\mathbb{F}_2$ as $\FF_2$. We will commonly use the notation $x = (v,w) \in \FF_2^{2n}$ where $v$ will always denote the first $n$ bits of $x$ and $w$ the second $n$ bits of $x$. We now state a few facts that we will use in this paper.
\begin{fact}
\label{fact:lowerboundexpectation}
    Let $Y$ be a random variable with $|Y|\leq 1$. If $\Exp[Y]\geq \varepsilon$, then $\Pr[Y\geq \delta]\geq (\varepsilon-\delta)/(1-\delta)$
\end{fact}
\begin{proof}
Let $B$ be a binary random variable. By definition of conditional expectation, we know that $\Exp[Y]=\sum_{b\in \01}\Exp[Y|B=b]\cdot \Pr[B=b]$. Now consider the indicator function $B=[Y\geq \delta]$. So we have that
\begin{align*}
\varepsilon\leq \Exp[Y]&=\Pr[Y\geq \delta]\cdot \Exp[Y|Y\geq \delta]+\Pr[Y< \delta]\cdot \Exp[Y|Y< \delta] \\
&\leq \Pr[Y\geq \delta]+(1-\Pr[Y\geq \delta])\cdot \delta\\
&= \delta+\Pr[Y\geq \delta] (1-\delta),
\end{align*}
where the inequality used that $ \Exp[Y|Y\geq \delta]\leq 1$ and $\Exp[Y|Y< \delta] \leq \delta$ and the fact that $\Exp[|Y|]\leq 1$. Hence we have that
 $\Pr[Y\geq \delta]\geq (\varepsilon-\delta)/(1-\delta)$
\end{proof}

\begin{lemma}[Chernoff bound]
\label{lem:chernoff}
Suppose we flip $T$ 0/1-valued coins, each taking value 1 with probability~$p$. Let $X$ be their sum (i.e., the number of 1s), which has expectation $\mu=pT$. The Chernoff bound implies $\Pr[X\leq(1-\delta)\mu]\leq \exp(-\delta^2\mu/2)$. 
\end{lemma}

\subsection{Fourier analysis}
We will work with complex-valued Boolean functions $f: \mathbb{F}_2^n \rightarrow \mathbb{C}$. The inner product (or correlation) of two functions $f, g: \mathbb{F}_2^n \rightarrow \mathbb{C}$ is given by
\begin{equation}
    \la f , g \ra = \Exp_x [f(x) \overline{g(x)}].
\end{equation}
Define the Fourier decomposition of $f$ as
$$
f(x)=\sum_S \widehat{f}(S)\chi_S(x),
$$
where $\chi_S(x)=(-1)^{S\cdot x}$ and the \emph{Fourier coefficients} $\widehat{f}(S)\in \CC$ are defined as
$$
\widehat{f}(S)=\Exp_x[f(x)\chi_S(x)].
$$
For a function $f:\01^n\rightarrow \CC$ and $v\in \01^n$, define the derivative operator as
$$
f_v(x)=f(x)\overline{f(v+x)}.
$$
We define the convolutions of two functions $f, g: \mathbb{F}_2^n \rightarrow \mathbb{C}$ as
\begin{equation}
    (f \star g)(x) = \Exp_{z \in \mathbb{F}_2^n} \left[f(x) g(x+z) \right].
\end{equation}
\paragraph{Symplectic Fourier transform.}
We will also work with the symplectic Fourier transform in addition to standard Fourier transform. Firstly, we introduce the symplectic inner product. For $x,y \in \mathbb{F}_2^{2n}$, where we write $x=(x_1, x_2)$ with $x_1$ denoting the first $n$ bits of $x$ and $x_2$ denoting the last $n$ bits (similarly for $y=(y_1,y_2)$), we have
\begin{equation}
    [x,y] = \la x_1, y_2 \ra + \la x_2, y_1 \ra \mod 2.
    \label{eq:symplectic_inner_product}
\end{equation}
The symplectic Fourier transform is defined as follows for functions $f: \FF_2^{2n} \rightarrow \R$:
\begin{align*}
    \widebreve{f}(a) = \frac{1}{4^n} \sum_{x \in \FF_2^{2n}} (-1)^{[a,x]} f(x), \quad
    f(x) = \sum_{a \in \FF_2^{2n}} (-1)^{[a,x]} \widebreve{f}(a)
\end{align*}    
Note that the standard Fourier coefficients are denoted by $\widehat{f}$ and symplectic Fourier coefficients by~$\widebreve{f}$. Usage will also be clear from context.

\subsection{Weyl operators, Lagrangian subspaces and key subroutines}\label{sec:weyl_ops}
The $2$-qubit Pauli matrices matrices are defined as follows
$$\id=\begin{pmatrix}
1 & 0\\
0 & 1
\end{pmatrix}, X=\begin{pmatrix}
0 & 1\\
1 & 0
\end{pmatrix}, Y=\begin{pmatrix}
0 & -i\\
i & 0
\end{pmatrix},Z=\begin{pmatrix}
1 & 0\\
0 & -1
\end{pmatrix}
$$
It is well-known that the $n$-qubit Pauli matrices $\{\id,X,Y,Z\}^n$ form an  {orthonormal basis} for $\mathbb{C}^n$.    In particular, for every $x=(a,b)\in \mathbb{F}_2^{2n}$, one can define the \emph{Weyl operator}
$$
W_x=i^{a\cdot b} (X^{a_1}Z^{b_1}\otimes X^{a_2}Z^{b_2} \otimes \cdots \otimes X^{a_n}Z^{b_n}).
$$  
and these operators $\{W_x\}_{x \in \FF_2^{2n}}$ are orthonormal. Note that each Weyl operator is a Pauli operator and indeed, every Pauli operator is a Weyl operator up to a phase. We define the commutation relations of Paulis $W_x,W_y$ for all $x,y \in \FF_2^{2n}$ using the symplectic product of their corresponding strings $x,y$. In particular, we say that two Paulis $W_x, W_y$ commute if $[x,y]=0$ and $W_x, W_y$ anti-commute if $[x,y]=0$. It is evident that $W_x W_y= (-1)^{[x,y]} W_y W_x$. Furthermore every   $\ket{\psi}$ can be written~as
$$
\ketbra{\psi}{\psi}=\frac{1}{2^n}\sum_{x\in \mathbb{F}_2^n} \alpha_x \cdot W_x,
$$
where   
$$
\alpha_x=\Tr(W_x \ketbra{\psi}{\psi}), \qquad \frac{1}{2^n}\sum_x \alpha_x^2=1.
$$   
Below we will use  {$p_\psi(x)=\alpha_x^2/2^n$}, so that $\sum_x p_\psi(x)=1$.    Since $n$-qubit Paulis can be associated with $2n$ bit strings, we will often refer to a Pauli $P$ by $x\in \01^{2n}$ by which we mean $P=W_x$.

\paragraph{Lagrangian subspace.} We say a subspace $S \subset \FF_2^{2n}$ is isotropic when $[x,y]=0$ for all $x,y \in S$ i.e., all the Weyl operators corresponding to the strings in $S$ commute with each other. We say that an isotropic subspace $S$ is a Lagrangian subspace when it is of maximal size $2^n$ i.e., $|S| = 2^n$. This ties in with the fact that a maximal set of $n$-qubit commuting Paulis is of size $2^n$.

\paragraph{Key subroutines.} We will feature two key routines in our work, namely, \emph{Bell sampling} and \emph{Bell difference sampling}. As its name suggests, Bell sampling $\ket{\psi} \otimes \ket{\phi}$ corresponds to measuring $\ket{\psi} \otimes \ket{\phi}$ in the Bell basis i.e., the orthonormal basis of $(W_x \otimes I) \ket{\Phi^{+}}$ with $\ket{\Phi^+}$ being the state of $n$ EPR pairs (over $2n$ qubits) $\ket{\Phi^+} := 2^{-n/2} \sum_{x \in \FF_2^n} \ket{x}\ket{x}$. The measurement outcome from Bell sampling is thus a $2n$ bit string $x \in \FF_2^{2n}$ that corresponds to a Weyl operator $W_x$. Notably, Montanaro~\cite{montanaro2017learning} showed that you could learn $n$-qubit stabilizer states using $O(n)$ samples from Bell sampling on $\ket{\psi}^{\otimes 2}$. Bell difference sampling a state $\ket{\psi}$ corresponds to Bell sampling on $\ket{\psi}^{\otimes 2}$ twice to produce outcomes $x,y \in \FF_2^{2n}$ and then returning $z = x + y$. Bell difference sampling was proposed in~\cite{gross2021schur} for intolerant testing stabilizer states.

\subsection{Characteristic and Weyl distributions}
We define the \emph{characteristic distribution}~as
\begin{equation}
    p_\Psi(x) = \frac{|\la \psi | W_x | \psi \ra|^2}{2^n},
\end{equation}
which satisfies $\sum_{x \in \FF_2^{2n}} p_\Psi(x)=1$. It can be observed that one can sample from the characteristic distribution by carrying out Bell sampling on $\ket{\psi} \otimes \ket{\psi^\star}$, where $\ket{\psi^\star}$ is the state's conjugate. The \emph{Weyl distribution}~\cite{gross2021schur} is defined as $q_\Psi$ as
\begin{equation}
    q_\Psi(x) = \sum_{y \in \mathbb{F}_2^{2n}} p_\Psi(y) p_\Psi(x+y).
    \label{eq:weyl_distribution}
\end{equation}
We can in fact sample from the Weyl distribution by simply carrying out Bell difference sampling on $4$ copies of $\ket{\psi}$ as was shown in~\cite{gross2021schur}, without requiring any access to $\ket{\psi^\star}$.

We now state some useful statements regarding the Weyl and characteristic distributions.

\begin{lemma}[Fact 3.2, \cite{grewal2022low}]
\label{lemma:expectation_paulis_qPsi}
    $\Exp_{x \sim q_\Psi}\left[|\la \psi | W_x | \psi \ra|^2 \right] = 2^{5n} \sum_{x \in \mathbb{F}_2^{2n}} \widehat{p}_\Psi^3(x)$.
\end{lemma}

\begin{proposition}[Proposition 3.3, \cite{grewal2022low}]
\label{prop:fourier_coeffs_pPsi}
    The Fourier coefficients of $p_\Psi(v,w)$ (where $v,w \in \mathbb{F}_2^n$) are $\widehat{p}_\Psi(v,w) = \frac{1}{2^n}p_\Psi(w,v)$.
\end{proposition}

\begin{proposition}[Theorem 3.2, \cite{gross2021schur}]
\label{prop:symplectic_fourier_coeffs_pPsi}
    For any $n$-qubit pure state $\ket{\psi}$, the symplectic Fourier coefficients of $p_\Psi(x)$ (where $x \in \mathbb{F}_2^{2n}$) are $\widebreve{p}_\Psi(x) = \frac{1}{2^n} p_\Psi(x)$.
\end{proposition}

\begin{lemma}\label{lemma:relation_expectation_state_derivative}
Let $W_x$ be a Weyl operator with $x=(v,w) \in \mathbb{F}_2^{2n}$ and $\ket{\psi}=\sum_x f(x)\ket{x}$. Then,
\begin{equation}
    \la \psi | W_x | \psi \ra = 2^n \widehat{f}_v (w)
\end{equation}
\end{lemma}
\begin{proof}
Consider the amplitude encoding of $\ket{\psi} = \sum_{x \in \mathbb{F}_2^n} f(x) \ket{x}$. We then have
\begin{align}
    \la \psi | W_x | \psi \ra = \sum_{x,y} i^{v \cdot w} \la y |\overline{f(y)} X^v Z^w f(x) |x \ra  &= i^{v \cdot w} \sum_{x,y} \la y | \overline{f(y+v)} (-1)^{w \cdot x} f(x) | x \ra \\
    &= i^{v \cdot w} \sum_{x,y} (-1)^{w \cdot x} f(x) \overline{f(y+v)} \la y | x \ra \\
    &= i^{v \cdot w} 2^n \Exp_x\left[(-1)^{w \cdot x} f(x) \overline{f(x+v)}\right] \\
    &= i^{v \cdot w} 2^n \widehat{f}_v(w),
\end{align}
where in the second line we used the fact that $\sum_{x \in \mathbb{F}_2^n} X^v f(x) \ket{x} = \sum_{x \in \mathbb{F}_2^n} f(x+v) \ket{x}$ and $Z^w \ket{x} = (-1)^{w \cdot x} \ket{x}$. The last equality follows from the definitions of the derivative of $f$ and Fourier coefficients.
\end{proof}

\begin{corollary}\label{corr:pPsi_derivatives}
    For $v,w \in \mathbb{F}_2^{n}$,  the characteristic distribution satisfies
  $p_{\Psi}(v,w) = 2^n |\widehat{f}_v(w)|^2
   $.
\end{corollary}

\begin{fact}
    For every valid quantum state $\ket{\psi}$ we have
    $$
\sum_{x,\alpha}p_\psi(x,\alpha)= 1, \quad  p_\psi(0,0)=2^{-n}, \qquad p_\psi(x,\alpha)\in [0,2^{-n}], \quad  |\textsf{supp}(p_\psi)|\geq 2^n
    $$
\end{fact}

\begin{proof}
First observe that $\ketbra{\psi}{\psi}=\frac{1}{2^n}\sum_{x,\alpha\in \mathbb{F}_2^n} \alpha_{x,\alpha} \cdot W_{x,\alpha}$
where   
$
\alpha_{x,\alpha}=\Tr(W_{x,\alpha} \ketbra{\psi}{\psi})$  and $\frac{1}{2^n}\sum_x \alpha_x^2=1.
$   
Now define  {$p_\psi(x,\alpha)=\alpha_{x,\alpha}^2/2^n$}, so this implies $\sum_{x,\alpha} p_\psi(x,\alpha)=1$. Furthermore, since $\Tr(\ketbra{\psi}{\psi})=1$, we have $\alpha_{0,0}=1$, so $p_\psi(0,0)=2^{-n}$ (since $\Tr(W_{x,\alpha})=2^n\cdot [x=\alpha=0]$. Finally its clear that $p_\psi(x,\alpha)\in [0,2^{-n}]$ and $|\textsf{supp}(p)|\cdot 2^{-n}\geq \sum_{x,\alpha}p_\psi(x,\alpha)=1$.
\end{proof}

\begin{lemma}
\label{lemma:expectation_paulis_qPsi_to_pPsi}
For every $\ket{\psi}$ we have that
\begin{equation*}
    \Exp_{x \sim q_\Psi}\left[|\la \psi | W_x | \psi \ra|^2 \right] = 2^{3n} \Exp_{y \in \mathbb{F}_2^n} \left[\sum_{\alpha \in \mathbb{F}_2^{n}} {p}_\Psi^3(y,\alpha)\right].
\end{equation*}    
\end{lemma}

\begin{proof}
Combine Lemma~\ref{lemma:expectation_paulis_qPsi} and Proposition~\ref{prop:fourier_coeffs_pPsi} to obtain
\begin{align*}
    \Exp_{x \sim q_\Psi}\left[|\la \psi | W_x | \psi \ra|^2 \right] = 2^{5n} \sum_{y,\alpha \in \mathbb{F}_2^{n}} \widehat{p}_\Psi^3(y,\alpha) &= 2^{2n} \sum_{y,\alpha \in \mathbb{F}_2^{n}} p_\Psi^3(y,\alpha)= 2^{3n} \Exp_{y \in \mathbb{F}_2^n} \left[\sum_{\alpha \in \mathbb{F}_2^n} p_\Psi^3(y,\alpha) \right].
\end{align*}
\end{proof}
\subsection{Stabilizer states and stabilizer groups}
 The output state of a  {Clifford circuit $C$} (i.e., a circuit consisting of  $H, S, CNOT$ gates)  on $\ket{0^n}$  {$C\ket{0^n}$} is referred to as a  \emph{stabilizer state}. An alternative description of stabilizer state $\ket{\psi}$ is as follows: it is the pure state for which there is a  {subgroup $\mathcal{S}\subseteq \{W_x\}$} of  {size $2^n$}    such that   $P$ stabilizes $\ket{\psi}$, i.e., $P\ket{\psi}=\ket{\psi}$    for all $P\in \mathcal{S}$,~i.e.,
$$
\ketbra{\psi}{\psi}=2^{-n}\sum_{\sigma\subseteq \mathcal{S}} \sigma,
$$   
where $\mathcal{S}\subseteq \{W_x\}_x$ has dimension $n$. Such a group is often referred to as a \emph{stabilizer group} (i.e., a set of $2^n$ Paulis that stabilize some stabilizer state). It is now not-too-hard to see that
\[
p_\psi(z)= 2^{-n}\cdot \langle \psi|W_z|\psi\rangle^2 =\begin{cases} 
      2^{-n} & z \text{ stabilizes } \ket{\psi} \\
      0 & \text{otherwise. } \\
   \end{cases}
\]   
If we did Bell difference sampling on copies of stabilizer states, then the output distribution satisfies
$$
{q_\psi(z)}=\sum_{a\in \mathbb{F}_2^{2n}}p_\psi(a)p_\psi(z+a)=   2^{-n}\sum_{a\in \textsf{Stab}(\psi)}p_\psi(z+a)=   2^{-n}[z\in \textsf{Stab}\ket{\psi}]= {p_\psi(z)}.
$$
Throughout this work, we will often use stabilizer groups (containing Weyl operators) and their corresponding Lagrangian subspaces (containing the $2n$ bit binary representations of the Weyl operators up to a phase) synonymously. The usage will be clear from context. 

Denote $\calF_\calS(\ket{\psi})$ as the maximum stabilizer fidelity of a quantum state $\ket{\psi}$. We then have the following facts regarding stabilizer fidelity.
\begin{fact}[\cite{gross2021schur,grewal2022low}]
\label{fact:lower_bound_stabilizer_fidelity}
Let $\ket{\psi}$ be an $n$-qubit pure quantum state. Then,
\begin{equation*}
  \Big(\Exp_{x \sim q_\Psi}\left[|\la \psi | W_x | \psi \ra|^2 \right]\Big)^{1/6}\geq   \calF_\calS(\ket{\psi}) \geq \frac{4}{3} \Exp_{x \sim q_\Psi}\left[|\la \psi | W_x | \psi \ra \right|^2] - \frac{1}{3}.
\end{equation*}
For the lower bound to be meaningful, one implicitly requires $\Exp_{x \sim q_\Psi}\left[|\la \psi | W_x | \psi \ra \right|^2] \geq 1/4$.
\end{fact}

\begin{fact}[Proof of Theorem 3.3, \cite{gross2021schur}; Corollary 7.4, \cite{grewal2023improved}]
\label{fact:lower_bound_stabilizer_fidelity_pPsi_lagrangian_subspace}
For any $n$-qubit quantum state $\ket{\psi}$ and a Lagrangian subspace $T \subset \FF_2^{2n}$
\begin{equation*}
    \calF_\calS(\ket{\psi}) \geq \sum_{x \in T} p_\Psi(x).
\end{equation*}
\end{fact}

For an arbitrary subset of Paulis $S \subseteq \{W_x\}_{x \in \FF_2^{2n}}$, we define a \emph{stabilizer covering} of $S$ as a set of stabilizer groups $\{G_1,\ldots,G_k\}$ such that $S \subseteq \bigcup_{i=1}^k G_i$. 

\section{Gowers norms}
\subsection{Definition of Gowers norm}
\paragraph{Gowers norm for Boolean functions}
For Boolean functions $f: \mathbb{F}_2^n \rightarrow \{-1,1\}$, the $U^{k}$ Gowers norm is defined as
\begin{equation}
    \|{f}\|^{2^{k}}_{U^{k}} = \mathop{\Exp}_{x,h_1,\ldots,h_k \in \FF_2^n} \Big[\prod_{S\subseteq [k]}f\big(x+\sum_{i\in S}h_i\big)\Big].
\end{equation}

\paragraph{Gowers norm for complex-valued functions}
For complex-valued functions defined on the Boolean cube $f: \mathbb{F}_2^n \rightarrow \mathbb{C}$, the $U^{k}$ Gowers norm is defined as
\begin{equation}
    \norm{f}_{U^{k}}^{2^{k}} = \Exp_{x,h_1,h_2,\ldots,h_k \in \mathbb{F}_2^n} \prod_{\omega \in \mathbb{F}_2^k} C^{|\omega|} f(x + \omega \cdot h),
\end{equation}
where $C^{|\omega|} f = f$ if $|\omega| := \sum_{j \in [k]} \omega_k$ is even and is $\overline{f}$ if $|\omega|$ is odd and $\omega\cdot h=(\omega_1h_1,\ldots,\omega_kh_k)$ .
A well-known result by Green and Tao~\cite{green2008inverse} is that the Gowers-$d$ norm measures how well a phase polynomial correlates with a degree-$d$ function. 
\begin{lemma}[\cite{green2008inverse}]
For every Boolean function $f: \mathbb{F}_2^n \rightarrow \mathbb{C}$ and a polynomial $p: \mathbb{F}_2^n \rightarrow \mathbb{F}_2^n$ of degree at most $d$, we have the following relation
\begin{equation}
    \left| \Exp_{x \in \mathbb{F}_2^n} \left[f(x) (-1)^{p(x)} \right] \right| \leq \norm{f}_{U^{d+1}}
\end{equation}
\end{lemma}

\paragraph{Gowers norm for quantum states.} For an arbitrary quantum state $\ket{\psi}=\sum_x f(x)\ket{x}$ where $f=(f(x))_x$ is an $\ell_2$-normed vector, we define its Gowers norm as follows  \begin{equation}
    \gowers{\ket{\psi}}{k}^{2^{k}} = 2^{n 2^{k-1}} \Exp_{x,h_1,h_2,\ldots,h_k \in \mathbb{F}_2^n} \prod_{\omega \in \mathbb{F}_2^k} C^{|\omega|} f(x + \omega \cdot h),
\end{equation}
where $C^{|\omega|} f = f$ if $|\omega| := \sum_{j \in [k]} \omega_k$ is even and is $\overline{f}$ if $|\omega|$ is odd. Note that the Gowers norm for quantum states also involves an extra pre-factor of $2^{n 2^{k-1}}$ in contrast to the usual definition for Boolean functions, to respect that $f$ is an $\ell_2$-normed vector.

\subsection{Properties of Gowers norm}
\begin{lemma} \label{lemma:gowers_3norm_boolean}
    For $f: \mathbb{F}_2^n \rightarrow \mathbb{C}$
    \begin{equation}
        \norm{f}_{U^3}^8 = \Exp_y \left[ \sum_{\alpha} |\widehat{f}_y(\alpha)|^4 \right].
    \end{equation}
\end{lemma}
\begin{proof}
    We first claim that
    \begin{align*}
        \norm{f}_{U^2}^4&= \Exp_{x,y,z} f(x) \overline{f(x+y)} \overline{f(x+z)} f(x+y+z)\\
        &=\Exp_{x,y,z}\sum_{S,T,U,V} \widehat{f}(S)\overline{\widehat{f}(T)}\overline{\widehat{f}(U)}\widehat{f}(V) \chi_S(x)\chi_T(x+y)\chi_U(x+z)\chi_V(x+y+z)
        =\sum_S|\widehat{f}(S)|^4,
    \end{align*} 
    using $\Exp[\chi_{S+T}(x)]=[S=T]$. Similarly, one can show that
    \begin{align*}
         \norm{f}_{U^3}^8
         &= \Exp_{x,y,z,w} f(x) \overline{f(x+y)} \overline{f(x+z)} \overline{f(x+w)} f(x+y+w)f(x+y+z)f(x+z+w) \overline{f(x+y+z+w)}\\
         &=\Exp_v\Exp_{x,y,z}\sum_{S,T,U,V} \widehat{f_v}(S)\overline{\widehat{f_v}(T)}\overline{\widehat{f_v}(U)}\widehat{f_v}(V) \chi_S(x)\chi_T(x+y)\chi_U(x+z)\chi_V(x+y+z)\\
         &=\Exp_v\left[ \sum_{S} |\widehat{f_v}(S)|^4 \right].
    \end{align*}
    These equalities conclude the proof of the lemma.
\end{proof}

We also have the following relation between $\gowers{\ket{\psi}}{3}^2$, which can now be interpreted as computing expectation values of Paulis sampled from the characteristic distribution of the state $\ket{\psi}$ (i.e., $\Exp_{x \sim p_\Psi}\left[|\la \psi | W_x | \psi \ra|^2 \right]$), with the expectation values of Paulis sampled from the Weyl distribution (i.e., $\Exp_{x \sim q_\Psi}\left[|\la \psi | W_x | \psi \ra|^2 \right]$).

\begin{lemma}[{Connecting Gowers norm and Weyl distribution}]
\label{lemma:relation_expectation_paulis_qPsi_and_pPsi} For every $\ket{\psi}$ we have that 
\begin{equation*}
    {\gowers{\ket{\psi}}{3}^{8}}=\Exp_{x \sim p_\Psi}\left[|\la \psi | W_x | \psi \ra|^2 \right]  \geq  \Exp_{x \sim q_\Psi}\left[|\la \psi | W_x | \psi \ra|^2 \right] \geq \underbrace{\left(\Exp_{x \sim p_\Psi}\left[|\la \psi | W_x | \psi \ra|^2 \right] \right)^2}_{= \gowers{\ket{\psi}}{3}^{16}}. 
\end{equation*}
\end{lemma}
\begin{proof}
We first prove the first equality.  Using Lemma~\ref{lemma:gowers_3norm_boolean}, we have
\begin{equation}
\label{eq:gowersintermsofPpsi}
    \gowers{\ket{\psi}}{3}^8 = 2^{4n} \Exp_y \left[ \sum_{\alpha} |\widehat{f}_y(\alpha)|^4 \right] = 2^{2n} \Exp_y \left[ \sum_{\alpha} p_\Psi^2(y,\alpha) \right].
\end{equation}
We expand the right hand side and obtain
    \begin{equation*}
        \Exp_{x \sim p_\Psi}\left[|\la \psi | W_x | \psi \ra|^2 \right] = 2^n \sum_{x \in \mathbb{F}_2^{2n}} p^2_\Psi(x) = 2^{2n} \Exp_{y \in \mathbb{F}_2^n} \left[\sum_{\alpha \in \mathbb{F}_2^{n}} {p}_\Psi^2(y,\alpha)\right],
    \end{equation*}
    where we substituted $x = (y,\alpha)$ in the last equality with $y,\alpha \in \mathbb{F}_2^n$. This proves the first equality.    
Next we first prove the first inequality. By the Cauchy-Shwartz inequality, we have
\begin{align}
    \sum_{y,\alpha \in \mathbb{F}_2^n} p_\Psi^2(y,\alpha) \leq \sqrt{\sum_{y,\alpha} p_\Psi(y,\alpha) \, \cdot \, \sum_{y,\alpha} p_\Psi^3(y,\alpha)} = \sqrt{\sum_{y,\alpha} p_\Psi^3(y,\alpha)},
\end{align}
where the first inequality used $(\sum_i \alpha_i^2)^2\leq \sum_i \alpha_i \cdot \sum_i \alpha_i^3$ by Cauchy-Schwartz inequality. 
Note we then have after applying Lemma~\ref{lemma:expectation_paulis_qPsi_to_pPsi} and the above expression
\begin{align}
\label{eq:relatingppsiandqpsi}
    \Exp_{x \sim q_\Psi}\left[|\la \psi | W_x | \psi \ra|^2 \right] = 2^{2n} \sum_{y,\alpha} p_\Psi^3(y,\alpha) \geq 2^{2n} \left( \sum_{y,\alpha} p_\Psi^2(y,\alpha)\right)^2 = \left(\Exp_{x \sim p_\Psi}\left[|\la \psi | W_x | \psi \ra|^2 \right] \right)^2,
\end{align}
where we used the definition of $\Exp_{x \sim p_\Psi}\left[|\la \psi | W_x | \psi \ra|^2 \right]$ in the last equality. The relation with Gowers norm for the first (left) inequality then follows from the first equality. For the second inequality, we observe that
\begin{align*}
\Exp_{x \sim q_\Psi}\left[|\la \psi | W_x | \psi \ra|^2 \right] = \sum_x q_\Psi(x) 2^n p_\Psi(x) &= \sum_x 2^n p_\Psi(x) \left[\sum_a p_\Psi(a) p_\Psi(x+a) \right] \\ 
& \leq \sum_x 2^n p_\Psi(x) \left[\sum_a p_\Psi(a)^2 \right] \\
&= 2^n \sum_x p_\Psi(x)^2 \\
&= \Exp_{x \sim p_\Psi}\left[|\la \psi | W_x | \psi \ra|^2 \right],
\end{align*}
where we have used Cauchy-Shwartz inequality in the first inequality. The relation with Gowers norm on the right inequality then follows from the first equality.
\end{proof}

We now show the following theorem which characterizes all quantum states with $\gowers{\ket{\psi}}{3}^8=~1$.
\begin{theorem}
    An $n$-qubit quantum state $\ket{\psi}$ has $\gowers{\ket{\psi}}{3}^8 = 1$ if and only if it is a stabilizer state.
\end{theorem}
\begin{proof}
Let us prove the forward case first. If $\ket{\psi}$ is a stabilizer state, then there is an Abelian group of Paulis $S$ of size $|S|=2^n$ whose elements (which are Paulis) stabilize the state $\ket{\psi}$. We will denote these Paulis by their corresponding $2n$-bit strings $(x,\alpha)\in \mathbb{F}_2^{2n}$. Then $p_\Psi(x,\alpha) = 2^{-n}$ for all $(x,\alpha) \in S$. Using Lemma~\ref{lemma:relation_expectation_paulis_qPsi_and_pPsi}, we then have that $\gowers{\ket{\psi}}{3}^8 = 1$. Now, suppose that $\ket{\psi}$ is a general pure quantum state satisfying $\gowers{\ket{\psi}}{3}^8 = 1$. Our goal is to then show that $\ket{\psi}$ is a stabilizer state. Using Lemma~\ref{lemma:relation_expectation_paulis_qPsi_and_pPsi}, we then have $\Exp_{x \sim q_\Psi}\left[|\la \psi | W_x | \psi \ra|^2 \right] \geq \gowers{\ket{\psi}}{3}^{16} = 1$. This in turn implies that $\calF_\calS(\ket{\psi}) = 1$ via Fact~\ref{fact:lower_bound_stabilizer_fidelity}. This completes the proof.
\end{proof}
We will now prove one side of the equality above for higher Gowers norms. To this end, we need the following well-known fact that shows that the output states of circuits consisting of gates from the diagonal Clifford heirarchy look like generalized phase states.
\begin{fact}[{\cite[Proposition~4]{arunachalam2022optimal}}]
\label{eq:factdiagonalleveld}
    Let $V$ be an $n$-qubit quantum circuit consisting of gates from the diagonal unitaries in the $d$-th level of the Clifford hierarchy. The state produced by the action of $V$ on $\ket{+}^{\otimes n}$ (up to a global phase) is
\begin{equation}
\label{eq:degreedphase}
    \ket{\psi_f} = \frac{1}{\sqrt{2^n}} \sum \limits_{x \in \01^n} \omega_q^{f(x)} \ket{x},
\end{equation}
where $q=2^d$, $\omega_q = \exp(2 \pi i/q)$ is the $q$'th root of unity, and the function $f:\mathbb{F}_2^n \rightarrow \ZZ_q$ with $\ZZ_q = \{0,1,\ldots,q-1\}$ being the ring of integers modulo $q$. The $\mathbb{F}_2$-representation of $f$ is as follows
\begin{equation}
    f(x) = \sum_{\substack{T\subseteq[n]\\ 1\le|T|\le d}}c_T\prod_{j\in T}x_j \pmod{2^d}, \quad c_T \in 2^{|T|-1} \ZZ_{2^{d+1-|T|}}.
    \label{eq:f2_representation_diagonal_unitaries}
\end{equation}
\end{fact}

Using this fact, the following theorem follows.
\begin{theorem}
Let $V$ be a circuit consisting of gates from the diagonal $d$-th level Clifford hierarchy. Then  $\gowers{V\ket{+}^n}{d+1}=1$.
\end{theorem}
\begin{proof}
By Fact~\ref{eq:factdiagonalleveld}, we know that $$V\ket{+}^n=\frac{1}{\sqrt{2^n}} \sum \limits_{x \in \01^n} \omega_q^{f(x)} \ket{x},$$
for some $f:\mathbb{F}_2^n \rightarrow \ZZ_q$. Now observe that
\begin{align*}
\gowers{\ket{\psi}}{d+1}^{2^{d+1}} &= \frac{2^{n2^d}}{\sqrt{2^{n2^{d+1}}}}\Exp_{x,h_1,h_2,\ldots,h_{d+1} \in \mathbb{F}_2^n} \prod_{\omega \in \mathbb{F}_2^{d+1}} C^{|\omega|} \omega_q^{f(x + \omega \cdot h)}\\
&= \Exp_{x,h_1,h_2,\ldots,h_{d+1} \in \mathbb{F}_2^n} C^{|\omega|} \omega_q^{\sum_{\omega \in \mathbb{F}_2^{d+1}} f(x + \omega \cdot h)}\\
&= \Exp_{x,h_1,h_2,\ldots,h_{d+1} \in \mathbb{F}_2^n}   C^{|\omega|} \omega_q^{(\nabla_{h_1,\ldots,h_{d+1}} f)(x)}\\
&= \Exp_{x,h_1,h_2,\ldots,h_{d+1} \in \mathbb{F}_2^n} \omega_q^{0}=1,
\end{align*}
where the third equality used that $\sum_{\omega \in \mathbb{F}_2^{d+1}} f(x + \omega \cdot h)$ equals the $(d+1)$th additive derivative of $f$ on input $x$, in the direction $h_1,\ldots,h_{d+1}$ (which we denote as $(\nabla_{h_1,\ldots,h_{d+1}}f)(x)$), and the fourth equality used that the $(d+1)$th  derivative of a degree $d$ function $f$ equals the constant $c=0$ by Fact~\ref{eq:factdiagonalleveld}.
\end{proof}
It is an interesting open question if one can prove some version of a converse to the theorem above, i.e., high Gowers-$(d+1)$ norm implies high correlation with some degree-$d$ generalized phase state. 

\subsection{Computing Gowers-\texorpdfstring{$3$}{3} norm using copies}
We now discuss how to compute ${\gowers{\ket{\psi}}{3}^{8}}$ given access to copies of the state $\ket{\psi}$ and its conjugate $\ket{\psi^\star}$. From Lemma~\ref{lemma:relation_expectation_paulis_qPsi_and_pPsi}, we have that $\gowers{\ket{\psi}}{3}^{8}=\Exp_{x \sim p_\Psi}\left[|\la \psi | W_x | \psi \ra|^2 \right]$. This implies that we can compute the ${\gowers{\ket{\psi}}{3}^{8}}$ norm by randomly sampling Paulis $W_x$ with respect to the characteristic distribution $p_\Psi$ of state $\ket{\psi}$ and then measuring $\ket{\psi}^{\otimes 2}$ in the basis $W_x^{\otimes 2}$. To access the characteristic distribution, we carry out Bell sampling on $\ket{\psi} \otimes \ket{\psi^\star}$ (previously described in Section~\ref{sec:weyl_ops}). The corresponding circuit is shown in Figure~\ref{fig:qcirc_gowers3_norm}.
\begin{figure}[h!]
    \begin{center}
    \begin{adjustbox}{scale=1}
    \begin{quantikz}
        \ket{\psi} & \ctrl{1} & \gate{H} & \meter{} & \setwiretype{c} \wire[l][1]["x"{above,pos=0.2}]{a}  & \ctrl[vertical wire=c]{1}  \\
        \ket{\psi^\star} & \targ{} & & \meter{} & \setwiretype{c} \wire[l][1]["\alpha"{above,pos=0.2}]{a}  & \ctrl[vertical wire=c]{2} \\
        \ket{\psi} &  & & & & \meterD{W_{x,\alpha}}\\
        \ket{\psi} &  & & & & \meterD{W_{x,\alpha}}
    \end{quantikz}
    \end{adjustbox}
    \end{center}
    \caption{Quantum circuit for estimating $\gowers{\ket{\psi}}{3}^{8}$ given access to copies of $\ket{\psi}$ and $\ket{\psi^\star}$.}
    \label{fig:qcirc_gowers3_norm}
\end{figure}

\begin{lemma}
\label{lemestimateU3}
Let $\ket{\psi}=\sum_{x\in \01^n} f(x)\ket{x}$. One can estimate $\gowers{\ket{\psi}}{3}^{8}$ up to additive error~$\delta$ using $O(1/\delta^2)$ copies of $\ket{\phi}$ and $O(n/\delta^2)$ total many one-qubit and two-qubit gates.
\end{lemma}

To avoid the use of copies of the conjugate of the quantum state $\ket{\psi^\star}$, we will also consider computing $\Exp_{x \sim q_\Psi}\left[|\la \psi | W_x | \psi \ra|^2 \right]$ which bounds the $\gowers{\ket{\psi}}{3}^{8}$ norm according to Lemma~\ref{lemma:relation_expectation_paulis_qPsi_and_pPsi}. We will see that this is also sufficient for the purposes of tolerant testing stabilizer states. We compute this expectation by performing repetitions of Bell difference sampling on $4$ copies of $\ket{\psi}$ to produce $W_{x+y}$ followed by measuring $\ket{\psi}^{\otimes 2}$ in the basis $W_{x+y}^{\otimes 2}$. The corresponding circuit is shown in Figure~\ref{fig:qcirc_exp_weyl}.
\begin{figure}[h!]
    \begin{center}
    \begin{adjustbox}{scale=1}
    \begin{quantikz}
        \ket{\psi} & \ctrl{1} & \gate{H} & \meter{} & \setwiretype{c} \wire[l][1]["x"{above,pos=0.2}]{a}  & \ctrl[vertical wire=c]{2}  \\
        \ket{\psi} & \targ{} & & \meter{} & \setwiretype{c} \wire[l][1]["\alpha"{above,pos=0.2}]{a}  & & \ctrl[vertical wire=c]{2}  \\
        \ket{\psi} & \ctrl{1} & \gate{H} & \meter{} & \setwiretype{c} \wire[l][1]["y"{above,pos=0.2}]{a}  & \targ{} & & \wire[l][1]["x+y"{above,pos=0.4}]{a} & \ctrl[vertical wire=c]{2} \\
        \ket{\psi} & \targ{} & & \meter{} & \setwiretype{c} \wire[l][1]["\beta"{above,pos=0.2}]{a}  &  & \targ{} & \wire[l][1]["\alpha+\beta"{above,pos=0.2}]{a} & \ctrl[vertical wire=c]{2} \\
        \ket{\psi} &  & & & & & & & \meterD{W_{x+y,\alpha+\beta}}\\
        \ket{\psi} &  & & & & & & & \meterD{W_{x+y,\alpha+\beta}}
    \end{quantikz}
    \end{adjustbox}
    \end{center}
    \caption{Quantum circuit for estimating $\Exp_{x \sim q_\Psi}\left[|\la \psi | W_x | \psi \ra|^2 \right]$ given access to copies of $\ket{\psi}$.}
    \label{fig:qcirc_exp_weyl}
\end{figure}

\begin{lemma}
\label{lem:estimate_exp_weyl}
Let $\ket{\psi}=\sum_{x\in \01^n} f(x)\ket{x}$. One can estimate $\Exp_{x \sim q_\Psi}\left[|\la \psi | W_x | \psi \ra|^2 \right]$ upto additive error~$\delta$ using $O(1/\delta^2)$ copies of $\ket{\phi}$ and $O(n/\delta^2)$ many gates.
\end{lemma}

\section{An inverse theorem for the Gowers-\texorpdfstring{$3$}{3} norm of quantum states}
\label{sec:inverse_theorem}
The main technical theorem we prove in this section is the following
\begin{restatable}{theorem}{gowersstates}\label{thm:inversegowersstates}
    Let $\gamma\in [0,1]$. If $\ket{\psi}$ is an $n$-qubit quantum state such that $\gowers{\ket{\psi}}{3}^8 \geq \gamma$, then there is an $n$-qubit stabilizer state $\ket{\phi}$ such that $|\la \psi | \phi \ra|^2 \geq \poly(\gamma)$.
    \end{restatable}
In order to prove this, we will prove the following three statements
\begin{enumerate}[$(i)$]
    \item If $\gowers{\ket{\psi}}{3}^8 \geq \gamma$, then there exists $S \subseteq \mathbb{F}_2^{2n}$ of size $|S|\in [\gamma/2 , 2/\gamma ]\cdot 2^n$ satisfying $(a)$ $S = \{(x,\alpha) : |\la \psi | X^x Z^{\alpha} | \psi \ra|^2 \geq \gamma/2\}$, and $(b)$ $\Pr_{(x,\alpha), (y,\beta) \in S}[ (x+y,\alpha+\beta) \in S]\geq \poly(\gamma)$. 
    \item If item $(i)$ holds true, there exists a subset $S' \subseteq S$ such that $|S'| \geq \poly(\gamma) |S|$ which is covered by few translates of a subgroup $V$ and $\Exp_{x \in V}\left[2^n p_\Psi(x) \right] \geq \poly(\gamma)$.
    \item If item $(ii)$ holds true, then we show $V$ has a stabilizer covering of $\poly(1/\gamma)$, which implies  that there is a stabilizer state $\ket{\phi}$ such that $|\la \psi | \phi \ra|^2 \geq \poly(\gamma)$.
    \end{enumerate}

As a side product, we will also show the following.
\begin{theorem}
\label{thm:inversegowersstates_weyl_expectations}
Let $\gamma\in [0,1]$, $C'>1$ be a constant. If $\ket{\psi}$ is an $n$-qubit quantum state such that $\Exp \limits_{x \sim q_\Psi}\left[|\la \psi | W_x | \psi \ra|^2 \right] \geq \gamma$, then there is an $n$-qubit stabilizer state $\ket{\phi}$ such that $|\la \psi | \phi \ra|^2 \geq \poly(\gamma)$.
\end{theorem}
We now prove the above listed items in the following sections below, leading to Theorem~\ref{thm:inversegowersstates}.

\subsection{Proving item \texorpdfstring{$1$}{1}: Finding an approximate subgroup}\label{sec:item1}
To prove this, we first need the following lemma connecting the Gowers norm and Weyl distribution.

\begin{lemma}
\label{lemma:high_gowers_implies_linearity_of_pPsi}
If $\gowers{\ket{\psi}}{3}^8 \geq \gamma$ then
\begin{equation*}
    2^{3n} \Exp_{x,y \in \mathbb{F}_2^n} \left[ \sum_{\alpha,\beta \in \mathbb{F}_2^n} p_\Psi(x,\alpha) p_\Psi(y,\beta) p_\Psi(x+y,\alpha+\beta) \right] \geq \gamma^2.
\end{equation*}
\end{lemma}
\begin{proof}
Using Lemma~\ref{lemma:relation_expectation_paulis_qPsi_and_pPsi}, we have that 
\begin{equation*}
    \Exp_{x \sim q_\Psi}\left[|\la \psi | W_x | \psi \ra|^2 \right]  \geq  \gowers{\ket{\psi}}{3}^{16} \geq \gamma^2. 
\end{equation*}
However, expectations of Paulis sampled from $q_\Psi$ can be shown to satisfy
\begin{align*}
    \gamma^2 \leq \Exp_{(x,\alpha) \sim q_\Psi}\left[|\la \psi | W_x | \psi \ra|^2 \right] 
    &= \sum_{(x,\alpha) \in \mathbb{F}_2^{2n}} q_\Psi(x,\alpha) 2^n p_\Psi(x,\alpha) \\
    &= 2^n \sum_{x,\alpha \in \mathbb{F}_2^n} \left[ \sum_{y,\beta \in \mathbb{F}_2^n} p_\Psi(y,\beta) p_\Psi(x+y,\alpha+\beta) \right] p_\Psi(x,\alpha) \\
    &= 2^{3n} \Exp_{x,\alpha \in \mathbb{F}_2^n} \left[ \sum_{y,\beta \in \mathbb{F}_2^n} p_\Psi(x,\alpha) p_\Psi(y,\beta) p_\Psi(x+y,\alpha+\beta) \right] 
\end{align*}
which concludes the lemma statement. 
\end{proof}

The following is then immediate.
\begin{corollary}\label{corr:relation_expectation_p_Psi_given_weyl_condn}
If $\Exp \limits_{x \sim q_\Psi}\left[|\la \psi | W_x | \psi \ra|^2 \right] \geq \gamma$.
\begin{equation}
    2^{3n} \Exp_{x,y \in \mathbb{F}_2^n} \left[ \sum_{\alpha,\beta \in \mathbb{F}_2^n} p_\Psi(x,\alpha) p_\Psi(y,\beta) p_\Psi(x+y,\alpha+\beta) \right] \geq \gamma.
\end{equation}
\end{corollary}
\begin{proof}
The proof can be seen as follows
\begin{align*}
    \gamma\leq  \Exp_{(x,\alpha) \sim q_\Psi}\left[|\la \psi | W_x | \psi \ra|^2 \right] 
    &= \sum_{(x,\alpha) \in \mathbb{F}_2^{2n}} q_\Psi(x,\alpha) 2^n p_\Psi(x,\alpha) \\
    &= 2^n \sum_{x,\alpha \in \mathbb{F}_2^n} \left[ \sum_{y,\beta \in \mathbb{F}_2^n} p_\Psi(y,\beta) p_\Psi(x+y,\alpha+\beta) \right]  p_\Psi(x,\alpha) \\
    &= 2^{3n} \Exp_{x,\alpha \in \mathbb{F}_2^n} \left[ \sum_{y,\beta \in \mathbb{F}_2^n} p_\Psi(x,\alpha) p_\Psi(y,\beta) p_\Psi(x+y,\alpha+\beta) \right],
\end{align*}
hence proving the corollary.
\end{proof}
\begin{theorem}
\label{thm:large_set_gamma}
Let $\gamma >0$. If $\gowers{\ket{\psi}}{3}^8 \geq \sqrt{\gamma}$, then there exists $S \subseteq \mathbb{F}_2^{2n}$ satisfying
\begin{enumerate}[(i)]
\item  $|S|\in [\gamma^2/80 , 4/\gamma ]\cdot 2^n$
    \item $S \subseteq \{(x,\alpha) : |\la \psi | X^x Z^{\alpha} | \psi \ra|^2 \geq \gamma/4\}$
    \item $\Pr_{(x,\alpha), (y,\beta) \in S}[ (x+y,\alpha+\beta) \in S]\geq \Omega(\gamma^5)$
    \item Additionally $0^{2n}$ lies in $S$.
\end{enumerate}
\end{theorem}
\begin{proof}
To prove this theorem, we first  define an intermediate set $X$ as follows:
\begin{equation}
    X = \{ (x,\alpha) : x, \alpha \in \mathbb{F}_2^{n} \, \text{ and } M(x,\alpha) \geq \gamma/4 \},
\end{equation}
where $M(x,\alpha)=2^np_\psi(x,\alpha)=|\langle \psi|W_{x,\alpha}|\psi\rangle|^2$.  By Lemma~\ref{lemma:relation_expectation_paulis_qPsi_and_pPsi} we have that
$$
\gowers{\ket{\psi}}{3}^8=\Exp \limits_{(x,\alpha) \sim p_\Psi}\left[|\la \psi | W_{x,\alpha} | \psi \ra|^2 \right] \geq \sqrt{\gamma}.
$$
Using Fact~\ref{fact:lowerboundexpectation} we have that
\begin{align*}
    \Exp \limits_{(x,\alpha) \sim p_\Psi}\left[|\la \psi | W_{x,\alpha} | \psi \ra|^2 \right] \geq \sqrt{\gamma}
    \implies 
    \Pr \limits_{(x,\alpha) \sim p_\Psi}\left[|\la \psi | W_{x,\alpha} | \psi \ra|^2 \geq \gamma/4 \right] \geq \gamma/2.
\end{align*}
We also have that
\begin{align*}
\gamma/2\leq \Pr \limits_{(x,\alpha) \sim p_\Psi}\left[|\la \psi | W_{x,\alpha} | \psi \ra|^2 \geq \gamma/4 \right] &= \sum \limits_{(x,\alpha) \in \01^{2n}}p_\psi(x,\alpha)\left[|\la \psi | W_{x,\alpha} | \psi \ra|^2 \geq \gamma/4 \right]\\ &\leq   2^{-n}\sum \limits_{(x,\alpha) \in \01^{2n}}\left[|\la \psi | W_{x,\alpha} | \psi \ra|^2 \geq \gamma/4 \right]\\
&= 2^{-n}\sum \limits_{(x,\alpha) \in \01^{2n}}\left[M(x,\alpha) \geq \gamma/4 \right]\\
&=2^{-n} \cdot |X|,
\end{align*}
where we used that $p_\psi(x,\alpha)\leq 2^{-n}$. Hence  the size of $X$ is at least $\frac{\gamma}{2} 2^n$. Now define a random set $X'\subseteq X$ (which we will eventually call $S$ in the theorem statement) as follows: $(\star)$ for every $(x,\alpha)\in X$, include $(x,\alpha)\in X'$ with probability $2^n p_\psi(x,\alpha)=M(x,\alpha)$ or discard it otherwise.\footnote{Note that $2^np_\psi(x,\alpha)\leq 1$ by definition.}  Also note that for every $\ket{\psi}$, we have that $\langle \psi|X^{0^n}Z^{0^n}|\psi\rangle=\langle \psi|\id|\psi\rangle=1$, so $M(0^{2n})=1$ and the $0^{2n}$ element is always included in $X'$ (proving item $3$ in the theorem statement). We now make three observations about this set $X'$.
\begin{enumerate}
    \item First the expected size of $|X'|$ is large, i.e.,
$$
\Exp[|X'|]=\sum_{(x,\alpha)\in X}2^n p_\psi(x,\alpha)=\sum_{(x,\alpha)\in X}M(x,\alpha)\geq \gamma/4\cdot |X|\geq \gamma^2/8 \cdot 2^n,
$$
where we used $M(x,\alpha)\geq \gamma/4$ for all $(x,\alpha)\in X$ and the lower bound of $|X|\geq \gamma/2\cdot 2^n$ that we showed above.
\item Second, by the randomized choice of the set $X'$ in $X$, the size of $|X'|$ is well concentrated around its mean $\mu=\Exp[|X'|]\geq \gamma^2/8\cdot 2^n$: By  Chernoff bound in Lemma~\ref{lem:chernoff} we have that
\begin{align}
\label{eq:concentrationofX'}
    \Pr\left[|X'|\leq \gamma^22^n/80\right]\leq \Pr\left[|X'|\leq \mu/10\right]
    &=\Pr\left[\sum_{(x,\alpha)\in X}[(x,\alpha) \in X']\leq \mu/10\right]\\
    &\leq \exp(-0.9^2\mu/2)\leq \exp(-0.4\gamma^2 2^n),
\end{align}
where the final inequality used that $\mu\geq \gamma^2/8\cdot 2^n$.
\item Third, the size of $X'$ isnt too big, i.e., 
  \begin{align}
  \label{eq:sizeofX'isnttoobig}
         |X'|\leq |X|=\sum_{x,\alpha}[M(x,\alpha)\geq \gamma/4]
=\sum_{x,\alpha}[p_\psi(x,\alpha)\geq 2^{-n}\cdot \gamma/4]\leq 2^n\cdot 4/\gamma,
        \end{align}
        where we used that  $\sum_{x,\alpha}p_\psi(x,\alpha)=1$.
\end{enumerate}
For every $Z\subseteq \01^{2n}$ define
\begin{equation}
    L(Z) = \Exp_{(x,\alpha), (y,\beta) \in Z} \left[ (x+y,\alpha+\beta) \in Z\right],
\end{equation}
where the probability is uniform over elements in $Z$. Note that for every $Z$, we have the following
  \begin{align}
        &\Exp_{(x,\alpha), (y,\beta) \sim p_\psi} [ (x+y,\alpha+\beta) \in Z]\cdot  [(x,\alpha) \in Z, (y,\beta) \in Z] \\
        &=\sum_{(x,\alpha), (y,\beta) }p_\psi(x,\alpha)p_\psi(y,\beta) [ (x+y,\alpha+\beta) \in Z]\cdot  [(x,\alpha) \in Z, (y,\beta) \in Z]\\
        &=\sum_{(x,\alpha), (y,\beta) \in Z }p_\psi(x,\alpha)p_\psi(y,\beta) [ (x+y,\alpha+\beta) \in Z]\\
         &\leq 2^{-2n}\cdot |Z|^2\cdot \Exp_{(x,\alpha), (y,\beta) \in Z}[ (x+y,\alpha+\beta) \in Z],
      \label{eq:understandingLlowerbound}
         \end{align}
        where we used $p_\psi(x,\alpha)\leq 2^{-n}$.
Next we lower bound $\Exp_{X'}[L(X')]$, for an $X'$ sampled from the distribution defined in $(\star)$, as follows:
\begin{align}
\Exp_{X'}[L(X')]&= \Exp_{X'}\Big[\Exp_{(x,\alpha), (y,\beta) \in X'}[ (x+y,\alpha+\beta) \in X']\Big]\\
&\geq2^{2n}\cdot \Exp_{X'} \left[|X'|^{-2}\cdot \mathop{\Exp}_{(x,\alpha), (y,\beta) \sim p_\psi} [ (x+y,\alpha+\beta) \in X']\cdot  [(x,\alpha) \in X', (y,\beta) \in X']\right] \\
&\geq \gamma^2/16\cdot \Exp_{X'} \left[\mathop{\Exp}_{(x,\alpha), (y,\beta) \sim p_\psi} [ (x+y,\alpha+\beta) \in X']\cdot  [(x,\alpha) \in X', (y,\beta) \in X']\right] \\
&= \gamma^2/16\cdot \Exp_{X'}\left[\sum_{(x,\alpha), (y,\beta)}p_\psi(x,\alpha)p_\psi(y,\beta)\cdot [ (x+y,\alpha+\beta),(x,\alpha),(y,\beta) \in X']\right]\\
&=\gamma^2/16\cdot \sum_{(x,\alpha), (y,\beta)}p_\psi(x,\alpha)p_\psi(y,\beta)\Pr_{X'}\left[ (x+y,\alpha+\beta),(x,\alpha),(y,\beta) \in X'\right]\\
&=\gamma^2/16\cdot \sum_{(x,\alpha), (y,\beta)}p_\psi(x,\alpha)p_\psi(y,\beta)M(x+y,\alpha+\beta)M(x,\alpha)M(y,\beta),
\end{align}
where the first inequality used Eq.~\eqref{eq:understandingLlowerbound}, second inequality used Eq.~\eqref{eq:sizeofX'isnttoobig}.    
It now remains to show that $\Exp_{X'} [L(X')] \geq \delta$, which we show indirectly below. Given that $\gowers{\ket{\psi}}{3}^8 \geq \sqrt{\gamma}$, Lemma~\ref{lemma:high_gowers_implies_linearity_of_pPsi} implies the following starting point from which we proceed 
\begin{align}
     \gamma &\leq 2^{3n} \Exp_{x,y \in \mathbb{F}_2^n} \left[ \sum_{\alpha,\beta \in \mathbb{F}_2^n} p_\Psi(x,\alpha) p_\Psi(y,\beta) p_\Psi(x+y,\alpha+\beta) \right] \label{eq:looseness_theorem}\\
     & = \sum_{(x,\alpha), (y,\beta) \in \01^{2n}} p_\Psi(x,\alpha) p_\Psi(y,\beta) M(x+y,\alpha+\beta)\\
    & = \sum_{(x,\alpha), (y,\beta) \in X} p_\Psi(x,\alpha) p_\Psi(y,\beta) M(x+y,\alpha+\beta) \nonumber\\
    &\qquad \qquad + 2^n \sum_{(x,\alpha) \text{ or } (y,\beta) \notin X} p_\Psi(x,\alpha) p_\Psi(y,\beta) p_\Psi(x+y,\alpha+\beta) \\
    & \leq  (16/\gamma^2)\cdot \sum_{(x,\alpha), (y,\beta) \in X} p_\Psi(x,\alpha) p_\Psi(y,\beta) M(x+y,\alpha+\beta)M(x,\alpha) M(y,\beta) \nonumber \\
    &\qquad \qquad + 2^n \sum_{(x,\alpha) \text{ or } (y,\beta) \notin X} p_\Psi(x,\alpha) p_\Psi(y,\beta) p_\Psi(x+y,\alpha+\beta) \\
    &\leq (16/\gamma^2)^2 \Exp_{X'}[L(X')] + 2^n \sum_{(x,\alpha) \text{ or } (y,\beta) \notin X} p_\Psi(x,\alpha) p_\Psi(y,\beta) p_\Psi(x+y,\alpha+\beta),
    \label{eq:sum_expectations_Lphi},
\end{align}
where the second inequality in line four used that for every $(x,\alpha)\in X$, we have that $M(x,\alpha)\geq \gamma/4$. 
We now bound the second term in the expression above. To this end, we first bound the sum over $(x,\alpha) \notin X$ i.e, $(x,\alpha) \in \mathbb{F}_2^{2n}$ for which $M(x,\alpha) < \gamma/4$.
\begin{align*}
    & 2^n \sum_{(x,\alpha) \notin X, (y,\beta) \in \mathbb{F}_2^{2n}} p_\Psi(x,\alpha) p_\Psi(y,\beta) p_\Psi(x+y,\alpha+\beta) \\
    & = \sum_{\substack{(x,\alpha) \in \mathbb{F}_2^n: \\ M(x,\alpha)<\gamma/4}} \sum_{(y,\beta) \in \mathbb{F}_2^{2n}} M(x,\alpha) p_\Psi(y,\beta) p_\Psi(x+y,\alpha+\beta) \\
    &\leq \gamma/4\cdot \sum_{(x,\alpha) \in \mathbb{F}_2^{2n}} \sum_{(y,\beta) \in \mathbb{F}_2^{2n}} p_\Psi(y,\beta) p_\Psi(x+y,\alpha+\beta)\\
&=\gamma/4\cdot\sum_{(y,\beta) \in \mathbb{F}_2^{2n}} p_\Psi(y,\beta) \sum_{(x,\alpha) \in \mathbb{F}_2^{2n}}  p_\Psi(x+y,\alpha+\beta)\\
&=\gamma/4\cdot\sum_{(y,\beta) \in \mathbb{F}_2^{2n}} p_\Psi(y,\beta) \sum_{(x',\alpha') \in \mathbb{F}_2^{2n}}  p_\Psi(x',\alpha')\\
&= \gamma/4,  
\end{align*}
where we used $\sum_{x,\alpha}p_\Psi(x,\alpha)=1$ above. 
We can similarly bound the sum over $(x,\alpha),(y,\beta)$ for which $M(y,\beta) < \gamma/4$ or $M(x+y,\alpha+\beta) < \gamma/4$. Combining with Eq.~\eqref{eq:sum_expectations_Lphi}, we thus have
\begin{align}
(16/\gamma^2)^2\Exp_{X'}\left[ L(X') \right] & \geq \gamma -3 \frac{\gamma}{4} \geq {\gamma}/4.
\end{align}
So we have that $\Exp_{X'}[L(X')]\geq \Omega(\gamma^5)$ and $\Exp_{X'}[|X'|/2^n]\geq \gamma^2/8$. It remains to argue that $\Pr_{X'}\Big[[|X'|/2^n \geq \gamma^2/8] \cap [L(X')\geq \Omega(\gamma^5)]\Big]>0$, which would imply the existence of an $X'$ satisfying the lemma statement. To this end observe that 
\begin{align}
    \gamma^5&\leq \Exp_{X'} [L(X')]\\
    &=\Exp[L(X')\Big||X'|/2^n\geq \gamma^2/80] \cdot \Pr[|X'|/2^n\geq \gamma^2/80] \\
    &\quad +\Exp[L(X')\Big||X'|/2^n\leq \gamma^2/80]\cdot \Pr[|X'|/2^n\leq \gamma^2/80]\cdot \\
    &\leq  \Exp[L(X')\Big||X'|/2^n\geq \gamma^2/80]+\Pr[|X'|/2^n\leq \gamma^2/80]\\
    &\leq  \Exp[L(X')\Big||X'|/2^n\geq \gamma^2/80]+\exp(-0.4\gamma^22^n),
\end{align}
where the first inequality used that both the probabilities and expectations are bounded by $1$
and the second inequality used Eq.~\eqref{eq:concentrationofX'}. The above inequalities implies that $\Exp[L(X')\Big||X'|/2^n\geq \gamma^2/80]\geq \gamma^5/10$. This now implies that
\begin{align*}
&\Pr_{X'}\Big[[|X'|/2^n \geq \gamma^2/80] \cap [L(X')\geq \Omega(\gamma^5)]\Big]\\
&=\Pr_{X'}[|X'|/2^n \geq \gamma^2/80]\cdot \Pr_{X'}\Big[ L(X')\geq \Omega(\gamma^5)    \Big\vert |X'|/2^n \geq \gamma^2/80\Big]\geq \poly(\gamma),
\end{align*}
hence there exists an $X'$ (which we will call $S$) satisfying the theorem statement. 
\end{proof}

\paragraph{Result for high expectation of Paulis from Weyl distribution.} Similarly, we can show the following theorem given that $\Exp \limits_{x \sim q_\Psi}\left[|\la \psi | W_x | \psi \ra|^2 \right]$ is high.
\begin{theorem}
\label{thm:large_set_gamma_qPsi_expectations}
Let $\gamma >0$. If $\Exp \limits_{x \sim q_\Psi}\left[|\la \psi | W_x | \psi \ra|^2 \right] \geq \gamma$, then there exists $S \subseteq \mathbb{F}_2^{2n}$  satisfying
\begin{enumerate}[(i)]
    \item $|S|\in [\gamma^2/80 , 4/\gamma ]\cdot 2^n$
    \item $S \subseteq \{(x,\alpha) : |\la \psi | X^x Z^{\alpha} | \psi \ra|^2 \geq \gamma/4\}$
    \item $\Pr_{(x,\alpha), (y,\beta) \in S}[ (x+y,\alpha+\beta) \in S]\geq \Omega(\gamma^5)$
    \item Additionally $0^{2n}$ lies in $S$.
\end{enumerate}
\end{theorem}
\emph{Remark.} The proof of Theorem~\ref{thm:large_set_gamma_qPsi_expectations} follows in the same exact manner as the proof of Theorem~\ref{thm:large_set_gamma}. The better dependence on $\gamma$ in the above statement is mainly due to the inequality in Eq.~\eqref{eq:looseness_theorem}. Application of Corollary~\ref{corr:relation_expectation_p_Psi_given_weyl_condn} here reduces the polynomial loss in $\gamma$ that is suffered when using Lemma~\ref{lemma:high_gowers_implies_linearity_of_pPsi} for Gowers-$3$ norm.

\paragraph{A corollary for mixed states.} We remark that one can prove an analogous statement as the theorem above also for mixed states. Let $\rho$ be an $n$-qubit quantum state and $\rho=\frac{1}{2^n}\sum_{x,\alpha}\alpha_{x,\alpha} W_{x,\alpha}$. Let $p_\rho(x,\alpha)=\alpha_{x,\alpha}^2/2^n$. It is not hard to see that \begin{equation}
    p_\rho(x,\alpha) = \frac{\Tr(\rho W_{x,\alpha})^2}{2^n}
\end{equation}
So we have that $\sum_{x,\alpha}p_\rho(x,\alpha)\leq 2^n$, and $p_\rho(x,\alpha)\leq 2^{-n}$.
\begin{theorem}
\label{thm:large_set_gamma_mixed_states}
Let $\gamma >0$. If 
$$
2^{3n} \Exp_{x,y \in \mathbb{F}_2^n} \left[ \sum_{\alpha,\beta \in \mathbb{F}_2^n} p_\rho(x,\alpha) p_\rho(y,\beta) p_\rho(x+y,\alpha+\beta) \right]  \geq \gamma,
$$ then there exists $S \subseteq \mathbb{F}_2^{2n}$ of size $|S|\in [\gamma^2/80 , 4/\gamma ]\cdot 2^n$
 satisfying
\begin{enumerate}[(i)]
    \item $S \subseteq  \{(x,\alpha) : |\Tr(\rho X^x Z^{\alpha} )|^2 \geq \gamma/4\}$
    \item $\Pr_{(x,\alpha), (y,\beta) \in S}[ (x+y,\alpha+\beta) \in S]\geq \Omega(\gamma^5)$
\end{enumerate}
\end{theorem}
We do not prove the theorem here since it is exactly the same proof as the proof of Theorem~\ref{thm:inversegowersstates}. 
\subsection{Proving item \texorpdfstring{$2$}{2}: Using results from additive combinatorics}\label{sec:item2}
In this section, we will proceed assuming that we are given that Gowers-$3$ norm is high and thus one can apply Theorem~\ref{thm:large_set_gamma} to obtain the set $S$ in its theorem statement. We could have proceeded in a similar manner given high $\Exp \limits_{x \sim q_\Psi}\left[|\la \psi | W_x | \psi \ra|^2 \right]$. In particular, note that the outputs of both Theorem~\ref{thm:large_set_gamma} and Theorem~\ref{thm:inversegowersstates_weyl_expectations} is a set $S$ which is approximately a group in $\FF_2^{2n}$. We now show that if $S$ is approximately a subgroup, then it can be covered by several translations of a single subgroup. To this end, we first require results in additive combinatorics namely the Balog-Szemeredi-Gowers (BSG) theorem~\cite{balog1994statistical,gowers2001new}.

\begin{theorem}[BSG Theorem]\label{thm:bsg}
    Let $G$ be an Abelian group and $S \subseteq G$. If 
    \begin{equation}
    \label{eq:requiredpromiseofS}
        \Pr_{s,s' \in S}[s + s' \in S] \geq \varepsilon,
    \end{equation}
    then there exists $S' \subseteq S$ of size $|S'| \geq (\varepsilon/3) \cdot |S|$ such that $|S' + S'| \leq (6/\varepsilon)^8 \cdot |S|$.
\end{theorem}
Note that for our purposes, $G = \FF_2^{2n}$ (or $\mathbb{F}_2^{2n}$) and is Abelian with respect to the operation of addition on $\FF_2$. Applying  Theorem~\ref{thm:bsg} on the set $S$ defined in Theorem~\ref{thm:large_set_gamma} (which satisfies Eq.~\eqref{eq:requiredpromiseofS} for $\varepsilon=\Omega(\gamma^5)$), there exists a set $S' \subseteq S$ and universal constants $c,c'>0$ such that  $|S'| \geq c\cdot  \gamma^{5} |S|$  and  $|S' + S'| \leq c' \gamma^{-40} |S|$. This implies that $|S' + S'| \leq c\cdot  c'\cdot  \gamma^{-35} |S'|$. So far, we have not established any other property of $S'$ except for that it is a subset of $S$ and satisfies $|S'+S'| \approx |S'|$.

\begin{remark}
\label{remark:nac}
Note that the $S$ by Theorem~\ref{thm:inversegowersstates} contains $\id$. If $S'$ does not contain $\id$, one can add $\id$ to $S'$. The new set $S'\cup \{\id\}$ has doubling constant at most $(6/\varepsilon^8)+2\varepsilon/3$ (for simplicity, below we ignore this additive $2\varepsilon/3$) and still satisfies $|S'\cap\{\id\}|\geq (\varepsilon/3) |S|$.
\end{remark}

To show that the set $S'$ identified above can be covered by translations of a group, we will use the recent breakthrough result of Gowers, Green, Manners and Tao~\cite{gowers2023conjecture,gowers2024marton} which resolved the polynomial Freiman–Ruzsa (also called the Marton's conjecture).
\begin{theorem}[{\cite[Theorem~1.1]{gowers2024marton}}]
\label{thm:marton_conjecture}
Suppose that $A \subseteq\mathbb{F}_2^{2n}$ is a set with $|A + A| \leq K |A|$. Then $A$ is covered by at most $2^8 K^{8}$ cosets of some subgroup $V \subset \textsf{span}(A)$ of size at most $|A|$.\footnote{We remark that the original statement from \cite{gowers2023conjecture} required $2K^{11}$ cosets but this was subsequently improved to $2K^9$ by Jyun-Jie Liao~\cite{liao2024improved} and then to $2^8 K^8$ by Gowers, Green, Manners and Tao~\cite{gowers2024marton}.}
\end{theorem}

Applying this theorem to the set $S'$ allows us to conclude that there exists a coset of $V\subseteq \mathbb{F}_2^{2n}$ with size $|V| \leq |S'|$, which we will denote by $z + V$ such that
\begin{equation}
    \left| S' \cap (z+V) \right| \geq \eta |S'| \geq \eta \cdot c\cdot \gamma^5|S|\geq \poly(\gamma) 2^n,
    \label{eq:intersection_S_coset_V}
\end{equation}
where $\eta = 1/(2^8 K^{8}) = C \cdot \gamma^{280}$ (for an absolute constant $C > 0$) and which we will denote simply as $\eta = C \poly(\gamma)$. We now argue that this implies $\Exp_{y \in V}[2^n p_\Psi(y)]$ is high. We do this by showing that $\left| S' \cap (z+V) \right|$ being high implies that the following expectation $\Exp_{y \in z + V}[2^n p_\Psi(y)]$ is high, which in turn is upper bounded by $\Exp_{y \in V}[2^n p_\Psi(y)]$. To this end, we first observe the following.
\begin{corollary}\label{corr:expectation_coset}
$\Exp_{y \in V}\left[ 2^n p_\Psi(y+z) \right] \geq \eta\cdot \gamma/4$.
\end{corollary}
\begin{proof}
The proof follows from the following inequalities
\begin{align*}
    \Exp_{y \in V}\left[ 2^n p_\Psi(y+z) \right] = \frac{2^n}{|V|} \cdot \sum_{y \in V}  p_\Psi(y+z) &\geq \frac{1}{|S'|} \sum_{y \in (z+V)} 2^n p_\Psi(y) \\
    &\geq \frac{1}{|S'|} \sum_{y \in S' \cap (z+V)} 2^n p_\Psi(y) \\
    &\geq \frac{1}{|S'|} |S' \cap (z+V)| \cdot {\gamma}/{2}\\
    &\geq  \eta\cdot \gamma/4,
\end{align*}
where we used $|V| \leq |S'|$ in the first inequality and used that $2^n p_\Psi(y) \geq \gamma/4$ for all $ y \in S' \subseteq S$, which is true by construction of $S$ in Theorem~\ref{thm:inversegowersstates}.
\end{proof}

\begin{claim}\label{claim:symplectic_ineq}
    Suppose $V$ is a subgroup. For any arbitrary $\alpha \in \FF_2^{2n}$, any $z \in V$ and $z' \notin V$, we have
    \begin{equation*}
        \sum_{y \in V} (-1)^{[\alpha,y + z]} \geq \sum_{y \in V} (-1)^{[\alpha,y + z']},
    \end{equation*}
    where $[\cdot,\cdot]$ is the symplectic inner product as we defined in Definition~\eqref{eq:symplectic_inner_product}.
\end{claim}
\begin{proof}
For an arbitrary $x \in \mathbb{F}_2^{2n}$, we will consider it as a concatenation of $x=(x_1,x_2)$ where $x_1$ is the string over the first $n$ bits of $x$, and $x_2$ is the string over the last $n$ bits of $x$. Let us simplify the left hand side of the expression in the above statement:
\begin{align}
 \Exp_{y \in V} \left[(-1)^{[\alpha,y + z]}\right] = \Exp_{y \in V} (-1)^{\la \alpha_1, y_2+z_2\ra + \la \alpha_2, y_1+z_1 \ra}
 &=\Exp_{y \in V} \left[(-1)^{ \la \alpha_1, y_2 \ra + \la \alpha_2, y_1 \ra }(-1)^{ \la \alpha_1, z_2 \ra + \la \alpha_2, z_1 \ra}\right]\\
 &=\Exp_{y \in V} \left[(-1)^{\la y, (\alpha_2, \alpha_1)\rangle} \right]  \cdot (-1)^{\la \alpha_1, z_2 \ra + \la \alpha_2, z_1 \ra}\\
 &=\mathds{1}[(\alpha_2,\alpha_1)\in V^{\perp}]\cdot  (-1)^{ \la \alpha_1, z_2 \ra + \la \alpha_2, z_1 \ra}\\
 &=\mathds{1}[(\alpha_2,\alpha_1)\in V^{\perp}]\cdot  (-1)^{\langle (z_1,z_2), (\alpha_2,\alpha_1)\rangle},
\end{align}
where we have denoted $V^{\perp} = \{w \in \mathbb{F}_2^{2n} : \la x, w \ra = 0 \,\forall x \in V\}$, and the inner product is over $\mathbb{F}_2$.
There are a few cases here: 
\begin{enumerate}
    \item If $(\alpha_2,\alpha_1) \notin V^\perp$, then the above expression both terms are $0$ and the claim follows.
    \item  If $(\alpha_2,\alpha_1) \in V^\perp$ and given $z \in V$, the above expression is $\Exp_{y \in V} \left[(-1)^{[\alpha,y + z]}\right] = 1$ since $\langle (z_1,z_2), (\alpha_2,\alpha_1)\rangle=0$. Since $z' \notin V$, we have two more cases:
    \begin{itemize}
        \item If $\la z', (\alpha_2,\alpha_1) \ra = 0$, then $\Exp_{y \in V} \left[(-1)^{[\alpha,y + z']}\right] = 1$.
        \item If $\la z', (\alpha_2,\alpha_1) \ra = 1$, then $\Exp_{y \in V} \left[(-1)^{[\alpha,y + z']}\right] = -1$.
    \end{itemize}
\end{enumerate}
Both these cases conclude the proof of the claim.
\end{proof}

\begin{claim}\label{claim:maximizing_expectation_coset}
For $z \in V$ and $z' \notin V$, we have
\begin{equation}
    \Exp_{y \in V}\left[ 2^n p_\Psi(y+z) \right] \geq \Exp_{y \in V}\left[ 2^n p_\Psi(y+z') \right]
\end{equation}
\end{claim}
\begin{proof}
Define $R_V(z) = 2^n \sum_{y \in V} p_\Psi(y + z)$. In order to prove this claim, we will work with the symplectic Fourier transform.
\begin{align*}
    \widebreve{R}_V(\alpha) = 2^n \Exp_{z \in \FF_2^{2n}} \left[\sum_{y \in V} p_\Psi(y+z) (-1)^{[\alpha,z]} \right]
    &= 2^{3n} \Exp_{z \in \FF_2^{2n}} \left[\sum_{y \in V} \Exp_{u \in \FF_2^{2n}} \left[ (-1)^{[u,y+z]} \widebreve{p}_\Psi(u)\right] (-1)^{[\alpha,z]} \right] \\
    &= 2^{3n} \sum_{y \in V} \Exp_{u,z \in \FF_2^{2n}} \left[ \widebreve{p}_\Psi(u) (-1)^{[u,y+z]+[\alpha,z]} \right] \\
    &= 2^{3n} \sum_{y \in V} \Exp_{u \in \FF_2^{2n}} \left[ \widebreve{p}_\Psi(u) (-1)^{[u,y]} \Exp_{z \in \FF_2^{2n}} \left[ (-1)^{[u + \alpha,z]} \right] \right] \\
    &= 2^{3n} \sum_{y \in V} \Exp_{u \in \FF_2^{2n}} \left[ \widebreve{p}_\Psi(u) (-1)^{[u,y]} \cdot \mathds{1}[u=\alpha] \right] \\
    &= 2^n p_\Psi(\alpha) \sum_{y \in V} (-1)^{[\alpha,y]},
\end{align*}
where in the fourth equality, we have used the bilinearity of the symplectic product as $[u,y] + [u,z] = [u,y+z]$. In the last equality, we used $\widebreve{p}_\Psi(x) = p_\Psi(x)/2^n$ (Prop.~\ref{prop:symplectic_fourier_coeffs_pPsi}). Note that we can then write
\begin{equation}
R_V(z) = \sum_{\alpha \in \FF_2^{2n}} \widebreve{R}_V(\alpha) (-1)^{[\alpha,z]} = 2^n \sum_{\alpha \in \FF_2^{2n}} p_\Psi(\alpha) \sum_{y \in V} (-1)^{[\alpha,y+z]}    
\label{eq:sym_fourier_series_R_V}
\end{equation}
We now use Claim~\ref{claim:symplectic_ineq} to conclude that $R_V(z) \geq R_V(z')$ for $z \in V$ and $z' \notin V$.
\end{proof}

Hence, Claim~\ref{claim:maximizing_expectation_coset} implies that $\Exp_{y \in V}\left[ 2^n p_\Psi(y+z) \right]$ is maximized for $z \in V$. From Corollary~\ref{corr:expectation_coset}, we thus have $\Exp_{y \in V}\left[ 2^n p_\Psi(y) \right] \geq \eta \cdot \gamma/4$.

\subsection{Proving item \texorpdfstring{$3$}{3}: Improved stabilizer covering}\label{sec:item3}
In this section, we will prove the following main theorem. 
\begin{theorem}\label{thm:stabilizer_covering_group}
Let $\ket{\psi}$ be an arbitrary pure quantum state and $V$ be a subgroup of $\mathbb{F}_2^{2n}$ such that it has high intersection with the set $S' \subseteq X$, i.e., 
$
\Exp_{y \in V}\left[ 2^n p_\Psi(y) \right] \geq \poly(\gamma).
$
Then, $V$ can be covered by a union of $O(\poly(1/\gamma))$ many stabilizer subgroups $G_j\subseteq \01^{2n}$.
\end{theorem}
\emph{Remark.} A consequence of the above theorem is that it motivates a conjecture in additive combinatorics that comments on the maximal size of a set of pairwise anti-commuting Paulis in sets with small doubling. This is described in Appendix~\ref{app_sec:conjecture}. We emphasize that this conjecture is not required as part of the proof and is included for the interested reader.

To do this, we need the following facts. In a slight abuse of notation, we will also use $V$ to denote the Weyl operators that correspond to the $2n$-bit strings in $V$. We use $X_i$ (or $Z_i$) to denote the Weyl operator with the single-qubit Pauli $X$ (or $Z$) acting on the $i$th qubit and trivially elsewhere. We use $\calP^k$ to denote the set of $k$-qubit Paulis. Moreover, we denote $\calP^m_Z=\{I,Z\}^{\otimes m}$. We write $I_\ell$ for the $\ell$-qubit identity operator.
\begin{fact}\label{fact:clifford_action_on_group}(\cite{fattal2004entanglement}) There exists $m+k \leq n$ and an $n$-qubit Clifford $U$~such~that
\begin{equation}
    U V U^\dagger = \la Z_1, X_1, \ldots, Z_k, X_k, Z_{k+1}, Z_{k+2}, \ldots, Z_{k+m} \ra.
\end{equation}
\end{fact}

Define the set $T := UVU^\dagger$. Using Fact~\ref{fact:clifford_action_on_group}, we have $T = UVU^\dagger = \calP^k \times \la Z_{k+1}, Z_{k+2}, \ldots, Z_{k+m}\rangle$.~Then
\begin{align*}
\frac{1}{|V|} \sum_{W_x \in \calP^k \times \la Z_{k+1},\ldots,Z_{k+m}\ra} 2^n p_{\tilde{\Psi}}(x) \geq \poly(\gamma)
\end{align*}
where the state $\tilde{\Psi}$ is the state obtained after application of the Clifford unitary $U$ i.e., $\tilde{\Psi} = U \Psi U^\dagger$. 

\begin{fact}
\label{eq:ub_sum1}
Let $\ell=n-k-m$.
We have the following
$$
\sum_{W_y \in \calP^k}
\sum_{W_z \in \calP^m_Z} \Tr\left( (W_y \otimes W_z \otimes I_\ell) \tilde{\Psi}\right)^2  \leq 2^{m+k} \,.
$$
\end{fact}
\begin{proof}
By definition, $\tilde{\Psi}$ is a projector onto a  normalized pure state. Thus
\[
 \Tr\left( (W_y \otimes W_z \otimes I_\ell) \tilde{\Psi}\right)^2  = \Tr\left( \tilde{\Psi} (W_y \otimes W_z \otimes I_\ell) \tilde{\Psi} (W_y \otimes W_z\otimes I_\ell)\right)
 \]
Since the Pauli group  is a unitary $1$-design, for any $k$-qubit operator $O$ one has $\sum_{W_y\in \calP^k} W_y O W_y = 2^k \Tr(O) \cdot I_k$.
As a consequence, 
\be
\sum_{W_y \in \calP^k} (W_y\otimes I_{m+\ell}) \tilde{\Psi}  (W_y\otimes I_{m+\ell}) =2^k I_k \otimes \rho
\ee
where  $\rho=\Tr_{1,\ldots,k} \tilde{\Psi}$ is an $(m+\ell)$-qubit density matrix obtained from $\tilde{\Psi}$ by tracing out the first $k$ qubits.
Accordingly, for any $W_z\in\calP^m_Z$ one has
\[
\sum_{W_y \in \calP^k} \Tr\left( (W_y \otimes W_z\otimes I_\ell) \tilde{\Psi}\right)^2 = 2^k \Tr\left( \rho (W_z\otimes I_\ell) \rho (W_z\otimes I_\ell)\right)\le
2^k\Tr(\rho)=
2^k.
\]
To get the last inequality we noted that  $\|(W_z \otimes I_\ell) \rho (W_z\otimes I_\ell)\|=\|\rho\|\le 1$.
We arrive at
\[
\sum_{W_y \in \calP^k}\sum_{W_z \in\calP^m_Z} \Tr\left( (W_y \otimes W_z\otimes I_\ell) \tilde{\Psi}\right)^2
= 2^k \sum_{W_z\in \calP^m_Z}  \Tr\left( \rho (W_z\otimes I_\ell) \rho (W_z\otimes I_\ell) \right) \le 2^{m+k}.
\]
\end{proof}
\emph{Remark.} We give another proof of the fact above in the appendix in Fact~\ref{app:proofoffact}, using first~principles.

\begin{claim}\label{claim:small_k}
$
k \leq O(\log (1/\gamma)).
$
\end{claim}
\begin{proof}
We noted earlier that $\Exp_{y \in V} [2^n p_\Psi(y)] \geq \poly(\gamma)$.
Since 
$\Exp_{y \in V} [2^n p_\Psi(y)] = \Exp_{y \in UVU^\dag} [2^n p_{\tilde{\Psi}}(y)]$, one gets
\begin{align*}
    \poly(\gamma) |V| &\leq \sum_{W_x \in \calP^k \times \la Z_{k+1},\ldots,Z_{k+m}\ra} 2^n p_{\tilde{\Psi}}(x) = \sum_{W_y \in \calP^k, W_z \in \calP_Z^m} \Tr\left( (W_y \otimes W_z \otimes I_\ell) \tilde{\Psi}\right)^2 
    \leq 2^{k+m}
\end{align*}
where the second inequality uses Fact~\ref{eq:ub_sum1}. Since $|V|=2^{2k+m}$, this we have $k=O(\log (1/\gamma))$.
\end{proof}

\begin{proof}[Proof of Theorem~\ref{thm:stabilizer_covering_group}.] We have previously noted that $T \equiv \calP^k \times \la Z_{k+1},\ldots,Z_{k+m} \ra$. Let us write the $4^k$ Paulis in $\calP^k = \{\tau_1,\tau_2,\ldots,\tau_{4^k}\}$. We choose $\tilde{S}_a = \la \tau_a \otimes I_{n-k} \ra \times \la Z_{k+1}, \ldots Z_{k+m} \ra \subseteq \calP^n$ for each $a \in [4^k]$. Moreover, we can extend $\tilde{S}_a$ to a stabilizer group $S_a$ on $n$-qubits. We then have 
\begin{equation}
    V \subseteq \bigcup \limits_{a=1}^{M} U^\dagger S_a U,
\end{equation}
where $M = 4^k$.~\footnote{In fact, we can produce a better stabilizer covering of size $2^k+1$ by considering the mutually unbiased bases that cover $\calP^k$~\cite{flammia2020efficient}.} Using Claim~\ref{claim:small_k}, we have $k \leq O(\log(1/\gamma))$ implying have that  $M=\poly(1/\gamma)$.~Moreover, the action of a Clifford $U$ on a stabilizer group $S_a$ is to take it to another stabilizer~group.
\end{proof}

\subsection{Putting everything together}\label{sec:proof_of_main_result}
We are now ready to prove our main theorem, restated below for convenience.
\gowersstates*

\begin{proof}

Using Theorem~\ref{thm:stabilizer_covering_group}, we can cover $V$ with a disjoint union of $M = O(\gamma^{-c})$ for some $c>1$ many Lagrangian subspaces $\{S_a\}_{a \in [M]}$ (also called stabilizer groups).
By item $(ii)$ we know that
$$
\frac{1}{|V|}\sum_{i \in [M]}\sum_{y\in V\cap S_i} 2^n p_\psi(y)\geq \frac{1}{|V|}\sum_{y\in V} 2^n p_\psi(y) \geq  \eta\cdot \gamma/4,
$$
hence there exists an $i^*\in [M]$ for which 
$$
\frac{1}{|V|}\sum_{y\in V\cap S_{i^*}} 2^n p_\psi(y)\geq \eta\cdot \gamma/(4M),
$$
and we can lower bound the size of $V$ using Eq.~\eqref{eq:intersection_S_coset_V} as $|V|\geq |S' \cap (z + V)| \geq \poly(\gamma) 2^n$. Hence, 
$$
\sum_{y\in V\cap S_{i^*}}  p_\psi(y) \geq \frac{|V|}{2^n} \cdot \eta \cdot \frac{\gamma}{4M} \geq \poly(\gamma) \cdot (4M)^{-1}  \geq \Omega(\poly(\gamma)) 
$$
Using Fact~\ref{fact:lower_bound_stabilizer_fidelity_pPsi_lagrangian_subspace} by considering the stabilizer group (or Lagrangian subspace) $S_{i^*}$, we have that
\begin{align*}
    \calF_\calS(\ket{\psi}) \geq \sum_{y \in S_{i^\star}} p_\Psi(y) \geq \sum_{y \in S_{i^*} \cap V} p_\psi(y) =\Omega(\poly(\gamma)).
\end{align*}
This concludes the proof.
\end{proof}

\begin{proof}[Proof of Theorem~\ref{thm:inversegowersstates_weyl_expectations}]
Lets quickly recap what we have shown so far and its implications when $\Exp \limits_{x \sim q_\Psi}\left[|\la \psi | W_x | \psi \ra|^2 \right] \geq \gamma$. Theorem~\ref{thm:large_set_gamma_qPsi_expectations} then gives us a set $S$ containing Paulis with high expectation values and is approximately a group. Particularly $|S|\in [\gamma^2/80 , 4/\gamma ]\cdot 2^n$ and $\Pr_{s, s' \in S}[ s + s' \in S]\geq \Omega(\gamma^5)$. Using the BSG Theorem (Theorem~\ref{thm:bsg}), we can show the existence of a set $S' \subseteq S$ such that  $|S'| \geq c\cdot  \gamma^{5} |S|$  and  $|S' + S'| \leq c' \gamma^{-40} |S|$ where $c,c'>0$ are universal constants. This implies that $|S' + S'| \leq c\cdot  c'\cdot  \gamma^{-35} |S'|$. This shows that $S'$ is a set with a small doubling, which we denote by $K =  (c'/c) \cdot \gamma^{-45}$. Using the polynomial Freiman-Ruzsa theorem (Theorem~\ref{thm:marton_conjecture}) combined with Claim~\ref{claim:maximizing_expectation_coset}, we obtain a subgroup $V$ of size $|V| \leq |S'|$ and $|V| \geq |V \cap (z + S')| \geq \eta |S'| \geq \eta \cdot c \cdot \gamma^5 |S| \geq \Omega(\eta \cdot \gamma^7) 2^n$ such that
$$
\Exp_{x \in V}\left[2^n p_\Psi(x)\right] \geq \eta \cdot \gamma/4,
$$
where $\eta = 1/(2^8 K^8) = \Omega(\gamma^{360})$ (where we used that $K=O(\gamma^{-45})$). At this point as in the proof of Theorem~\ref{thm:inversegowersstates}, we comment on the size of the stabilizer covering of $V$. From Theorem~\ref{thm:stabilizer_covering_group} and using mutually unbiased bases to construct the stabilizer covering as we had noted could be done earlier, we obtain a stabilizer covering of $V$ of size $M = 2^k + 1 \leq O(1/(\eta \cdot \gamma)) \leq O(\gamma^{-361})$ (where $k$ has the same meaning as in Section~\ref{sec:item3}). We can now completely proceed as was done in the proof of Theorem~\ref{thm:inversegowersstates} to conclude that there exists a stabilizer state $\ket{\phi}$ such that
\begin{align*}
    |\la \psi | \phi \ra|^2 \geq \sum_{y\in V\cap S_{i^*}}  p_\psi(y) \geq \frac{|V|}{2^n} \cdot \eta \cdot \frac{\gamma}{4M} \geq \cdot \eta^2 \cdot \frac{\gamma^8}{4M} \geq \Omega(\gamma^{728}) \cdot (4M)^{-1}  \geq \Omega(\gamma^{1089}),
\end{align*}
where in the third inequality we used that $|V| \geq \Omega(\eta \cdot \gamma^7) 2^n$ and substitute for $\eta = \Omega(\gamma^{360})$. This proves our theorem statement.
\end{proof}

\section{Tolerant testing algorithm}
We are now finally ready to prove our main theorem on tolerant testing stabilizer states. We have the ability to estimate both the Gowers-$3$ norm using Bell sampling (Lemma~\ref{lemestimateU3}) and $\Exp \limits_{x \sim q_\Psi}\left[|\la \psi | W_x | \psi \ra|^2 \right]$ using Bell difference sampling (Lemma~\ref{lem:estimate_exp_weyl}) at our disposal. The two resulting tolerant testing algorithms only differ in the access required. Here, we describe guarantees for both testing algorithms but our main showcase is the testing algorithm that only requires access to copies of $\ket{\psi}$.

\begin{theorem}
\label{thm:test_stab_with_gowers}
Let $\varepsilon_1>0$ and $\varepsilon_2\leq \poly(\varepsilon_1)$.  There is an algorithm that~given
$
\poly( 1/\varepsilon_1)
$
copies of an $n$-qubit $\ket{\psi}$ and its conjugate $\ket{\psi^\star}$, can decide if $\max_{\ket{\phi}\in \Sh}|\langle\phi |\psi\rangle|^2 \geq  \varepsilon_1$ or $\max_{\ket{\phi}\in \Sh}|\langle\phi |\psi\rangle|^2 \leq  \varepsilon_2$ with gate complexity $O(n\cdot \poly(1/\varepsilon_1))$.
\end{theorem}
\begin{proof}
Let $\delta$ be a parameter that we fix at the end. The testing algorithm simply takes $O(1/\delta^2)$ copies of $\ket{\psi},\ket{\psi^*}$ and estimates $2^{4n}\cdot \gowers{\ket{\psi}}{3}^8$ up to additive error $\delta/2$. Observe that in the \emph{no} instance (i.e.,  $\max_{\phi\in \Sh}|\langle\phi |\psi\rangle|^2 \leq  \varepsilon_2$),  Theorem~\ref{thm:inversegowersstates} implies that 
$$
\gowers{\ket{\psi}}{3}^8\leq {\frac{1}{ \poly(\varepsilon_2)}}.
$$ 
In the \emph{yes} instance  (i.e.,  $\max_{\psi\in \Sh}|\langle\phi |\psi\rangle|^2 \geq  \varepsilon_1$), combining Lemma~\ref{lemma:relation_expectation_paulis_qPsi_and_pPsi} with Theorem~\ref{fact:lower_bound_stabilizer_fidelity} implies
$$
\gowers{\ket{\psi}}{3}^8\geq  \varepsilon_1^{6}.
$$
By letting 
$$
\delta=\frac{1}{10}\Big(\varepsilon_1^{6}-{{\poly(\varepsilon_2)}}\Big),
$$
the testing algorithm can distinguish between the \emph{yes} and \emph{no} instances of the problem using 
\begin{align*}
\poly(1/\delta)=\poly\Big({\varepsilon_1^{6}-{{\poly (\varepsilon_2)}}}\Big)^{-1}=\poly(1/\varepsilon_1)
\end{align*}
many copies of $\ket{\psi}$ and $\ket{\psi^\star}$~(Lemma~\ref{lemestimateU3}), where we used that $\varepsilon_2\leq \poly(\varepsilon_1)$. 
The gate complexity of this protocol involves a factor-$n$ overhead in comparison to the sample complexity.
\end{proof}

Now, for our main result which we restate for the reader's convenience and is proved very similarly to the above theorem.
\tolerantstabtesting*
\begin{proof}
Let $\delta$ be a parameter that we fix at the end. The testing algorithm simply takes $O(1/\delta^2)$ copies of $\ket{\psi}$ and estimates $\Exp \limits_{x \sim q_\Psi}\left[|\la \psi | W_x | \psi \ra|^2 \right]$ up to additive error $\delta/2$. Observe that in the \emph{no} instance (i.e.,  $\max_{\phi\in \Sh}|\langle\phi |\psi\rangle|^2 \leq  \varepsilon_2$),  Theorem~\ref{thm:inversegowersstates_weyl_expectations} implies that 
$$
\Exp \limits_{x \sim q_\Psi}\left[|\la \psi | W_x | \psi \ra|^2 \right] \leq {\frac{1}{\poly (\varepsilon_2)}}.
$$ 
In the \emph{yes} instance  (i.e.,  $\max_{\psi\in \Sh}|\langle\phi |\psi\rangle|^2 \geq  \varepsilon_1$), Theorem~\ref{fact:lower_bound_stabilizer_fidelity} implies that
$$
\Exp \limits_{x \sim q_\Psi}\left[|\la \psi | W_x | \psi \ra|^2 \right] \geq  \varepsilon_1^{6}.
$$
By letting 
$$
\delta=\frac{1}{10}\Big(\varepsilon_1^{6}-{{\poly(\varepsilon_2)}}\Big),
$$
the testing algorithm can distinguish between the \emph{yes} and \emph{no} instances of the problem using 
\begin{align*}
\poly(1/\delta)=\poly\Big({\varepsilon_1^{6}-{{\poly (\varepsilon_2)}}}\Big)^{-1}=\poly(1/\varepsilon_1)
\end{align*}
many copies of the unknown state~(Lemma~\ref{lem:estimate_exp_weyl}), where we used that $\varepsilon_2\leq \poly(\varepsilon_1)$. The gate complexity of this protocol involves a factor-$n$ overhead in comparison to the sample complexity.
\end{proof}

%%%%%%%%%%%%%%%%%%%%%%%%%%%%%%%%%%%%%%%%%%%%%%%%%%%%%%%%%%%%%%%

\bibliographystyle{alpha}
\bibliography{references}
%%%%%%%%%%%%%%%%%%%%%%%%%%%%%%%%%%%%%%%%%%%%%%%%%%%%%%%%%%%%%%%
\appendix

\section{A quantum-inspired conjecture in additive combinatorics}\label{app_sec:conjecture}
We begin this section of the Appendix by first defining some notation. Denote the maximal set of mutually pairwise anti-commuting Paulis (corresponding to their $2n$ bit strings) in a set $A$ as $\AC(A)$ (i.e., for every $P,Q$ in $A$, we have $[P,Q]=1$), and its size as $\nac(A) = |\AC(A)|$. For sets $A_1,A_2 \subseteq \01^n$, define
$$
A_1+A_2=\{a_1+a_2:a_1\in A_1, a_2\in A_2\}.
$$ 
We also define the $t$ sumset of $A$ denoted by $tA$ as
\begin{equation}
    tA := \{a_1 + a_2 + \ldots a_t : a_1, a_2, \ldots, a_t \in A\}.
\end{equation}
For a set $A\subseteq \01^n$, we saw $A$ is covered by \emph{translates} of $B$ if, there exists $k\geq 1$ and $c_1,\ldots,c_k\in \01^n$ such that 
$$
A\subseteq \cup_{i\in [k]} (c_i+B),
$$
where $c+B=\{c+b:b\in B\}$. In this section, we will describe how the proof of the inverse theorem of the Gowers-$3$ norm of quantum states (Theorem~\ref{thm:inversegowersstates}) inspired the conjecture in additive combinatorics~(Conjecture~\ref{conjecture:doubling}), which we restate here.
\begin{conjecture}
\label{conjecture:doubling_restated}
Let $S \subseteq \01^{2n}$ be s.t.~$|S|\geq 2^n/K$ and $|2S|\leq K |S|$. Then, $\nac(2S)\leq \poly(K,\nac(S))$.
\end{conjecture} 
We will now describe the origin of the above conjecture and give credence to why it may be true.

\subsection*{Bounds on number of anti-commuting Paulis}\label{app_sec:bounds_nac_V}
The goal in this section is to bound the number of anti-commuting Paulis in $V$.  Below we denote $K=\poly(1/\gamma)$. We start off with the following theorem, which gives a non-trivial bound on $\nac(A)$ if we are promised that $A$ is covered by a few translations of a subset $B$.

\begin{theorem}
\label{thm:max_AC_set_with_small_covering}
If $A \subseteq \mathbb{F}_2^{2n}$ is covered by at most $M$ translates of $B \subseteq \mathbb{F}_2^{2n}$, then
\begin{equation}
    \nac(A) \leq 2M\cdot  \nac(B).
\end{equation}
\end{theorem}
\begin{proof}
Let us denote the set of translates as $X$, which satisfies $|X| \leq M$ by assumption. We are given that one can cover $A$ by translates in $X$ of $B$. In other words, $A \subseteq X + B$. We thus have~that
\begin{align*}
\AC(A) \subseteq \bigcup \limits_{P \in X} (\AC(A) \cap \AC(P\cdot B)) \subseteq \bigcup \limits_{P \in X} \AC(P\cdot B).
\end{align*}
For each $P \in X$, let $m(P)$ be the size of maximal set of pairwise anti-commuting Paulis in $P\cdot B := \{PQ | Q \in B\}$. Clearly, we then have
\begin{equation}
    \nac(A) \leq \sum_{P \in X} m(P) \leq |X| \cdot \max_{P \in X} m(P).
    \label{eq:relation_nac_A_to_B}
\end{equation}
We now claim that
\begin{equation}
    m(P) \leq 2 \nac(B)
\end{equation}
for any $P \in X$. Let $m \equiv m(P)$ and suppose $PQ_1,\ldots, PQ_m$ be pairwise anti-commuting Paulis, where $Q_1\ldots,Q_m \in B$. Assume wlog that $Q_1,\ldots,Q_k$ anti-commute with $P$ and $Q_{k+1},\ldots,Q_m$ commute with $P$. Let $[P_1,P_2]$ denote the symplectic inner product between two Paulis $P_1$ and $P_2$. We then have
\begin{equation}
    [PQ_i,PQ_j] = [P,Q_j] + [P,Q_i] + [Q_i,Q_j] \mod 2 \, , \forall i \neq j \in [m]
    \label{eq:condition_anticommuting}
\end{equation}
If $i,j \in [k]$, then $[P,Q_i] = [P,Q_j]=1$. From Eq.~\eqref{eq:condition_anticommuting}, we are then left with $[Q_i,Q_j] = [PQ_i,PQ_j] = 1, \forall i \neq j \in [k]$ as by assumption $PQ_i$ and $PQ_j$ anti-commute. Similarly, if $k+1 \leq i,j \leq m$ then $[P,Q_i] = [P,Q_j] = 1$. Thus, Eq.~\eqref{eq:condition_anticommuting} gives us $[Q_i,Q_j] = [PQ_i,PQ_j] = 1, \forall  k+1 \leq i,j \leq m$ as by again assumption $PQ_i$ and $PQ_j$ anti-commute. Note that when $1 \leq i \leq k$ and $k+1 \leq j \leq m$, we have $[P,Q_i] = 1$ but $[P,Q_j]=0$. Eq.~\eqref{eq:condition_anticommuting} then gives us that $[Q_i,Q_j] = 0$ in this case. It then follows that
\begin{equation}
    \nac(B) \geq \max\{k,m-k\} \geq \frac{m}{2} = \frac{m(P)}{2},
\end{equation}
hence proving the claim. Substituting the claim in Eq.~\eqref{eq:relation_nac_A_to_B} gives us the desired result.
\end{proof}
This implies that the number of anti-commuting Paulis in $4S'$ is not much larger than in~$S'$.
\begin{corollary}
\label{corr:small_doubling_max_AC_set_S}
If $S' \subseteq \mathbb{F}_2^{2n}$ satisfies $|S'+S'| \leq K |S'|$, then
$
\nac(4S') \leq K^5 \nac(2S').
$    
\end{corollary}

To show the above statement, we will need the following lemmas.
\begin{lemma}[Corollary~2.20,~\cite{taovu2006comb}]
\label{lemma:size_of_4A}
    For all sets $A \subseteq \FF_2^{2n}$ such that $|2A|/|A| \leq K$, we have that 
    $$
    |4A|/|2A| \leq (|2A|/|A|)^4\leq K^4.
    $$
\end{lemma}

\begin{proof}[Proof of Corollary~\ref{corr:small_doubling_max_AC_set_S}]
Using the lemma above, we can cover $4S'$ by at most $K^5$ translates of $2S'$, where the set of translates $X \subseteq 3S'$ such that $|X| \leq K^5$. In other words, $4S' \subseteq X + 2S'$. At this point, we use Theorem~\ref{thm:max_AC_set_with_small_covering} to obtain the desired result
\begin{equation*}
\nac(4S') \leq |X| \nac(S') \leq K^5 \nac(2S'),
\end{equation*}
hence proving the corollary.
\end{proof} 

\begin{lemma}[Ruzsa's covering lemma~\cite{taovu2006comb}]
\label{lemma:rusza_covering}
For $A,B \in \FF_2^{2n}$, there exists a set $X \subseteq B$ such~that
$$
B \subseteq X + 2A; \quad |X| \leq \frac{|A+B|}{|A|}; \quad |A + X| = |A| \cdot |X|.
$$
In particular, $B$ can be covered by at most $|A+B|/|A|$ translates of $2A$.
\end{lemma}
Using this lemma, we are able to show that  the number of anticommuting Paulis in $V$ is~bounded.
\begin{corollary}
\label{corr:nac_V_nac_S}
$$
\nac(V) \leq \poly(K) \nac(2S')
$$    
\end{corollary}
\begin{proof}
From Theorem~\ref{thm:marton_conjecture}, we have that $S'$ can be covered by $O(K^8)$ translates (or cosets) of the subgroup $V$. Let us denote this set of translates as $X$. We then have that $S' \subseteq X + V$. Invoking Ruzsa's covering lemma (Lemma~\ref{lemma:rusza_covering}) with $B \equiv V$ and $A \equiv S'$ gives us that $V$ itself can be covered by at most $M$ translates of $2S'$ where
\begin{equation*}
    M = \frac{|S' + V|}{|V|} \leq \frac{|X + V + V|}{|V|} \leq \frac{|X + V|}{|V|} \leq \frac{|X||V|}{|V|} = |X| \leq O(K^8).
\end{equation*}
Let us call these translates as $Y$. We then have $V \subseteq Y + 2S'$ where $|Y| = M \leq O(K^8)$. Now, invoking Theorem~\ref{thm:max_AC_set_with_small_covering} gives us that
\begin{equation*}
    \nac(V) \leq O(K^8) \cdot\nac(2S'),
\end{equation*}
proving the corollary.
\end{proof}

\subsection*{Stabilizer covering and connection to conjecture.}\label{sec:item4}
We will now prove the following theorem.
\begin{theorem}\label{app_thm:stabilizer_covering_V_combinatorial}
Let $\ket{\psi}$ be an $n$-qubit pure quantum state and let $V$ be a subgroup of $\mathbb{F}_2^{2n}$ such that it satisfies
$$
\Exp_{y \in V}\left[ 2^n p_\Psi(y) \right] \geq \eta \cdot \gamma/4,
$$
and $\eta= C \cdot \poly(\gamma)$ for some constant $C > 0$. Then, $V$ can be covered by a union of at most $O\big(\exp(\nac(2S')\cdot \gamma^{-c})\big)$ many Lagrangian subspaces for some constant $c > 0$. In other words, the set of Weyl operators corresponding to $V$ has a stabilizer covering of size at most $O(\exp(\nac(2S')\cdot \gamma^{-c}))$
\end{theorem}

To do this, we need the following facts. In a slight abuse of notation, we will use $V$ to also denote the Weyl operators that correspond to the $2n$-bit strings in $V$. Let us also denote the $n$-qubit Pauli group (i.e., $\{W_x\}_{x \in \FF_2^{2n}}$) as $\calP^n$.
\begin{fact}\label{fact:clifford_action_on_group}
Suppose $V \subseteq \calP^n$ is a group. Then, there exists integers $m,k \geq 0$ such that $m+k \leq n$ and an $n$-qubit Clifford unitary $U$ such that
\begin{equation}
    U V U^\dagger = \la Z_1, X_1, \ldots, Z_k, X_k, Z_{k+1}, Z_{k+2}, \ldots, Z_{k+m} \ra,
\end{equation}
where by $Z_j$ (resp.~$X_j$) we mean the $n$-fold Pauli product of identity on every qubit except for the $j$th qubit which takes a $Z$ Pauli (resp.~$X$), and $\langle \{A_i\}_{i\in [k]}\rangle=\{\prod_{i\in S}A_i\}_{S\subseteq [k]}$.
\end{fact}

\begin{fact}[\cite{sarkar2021sets}]\label{fact:size_set_anticommuting_paulis}
The $k$-qubit Pauli group $\calP^k$ contains $(2k+1)$ anti-commuting Paulis.
\end{fact}

Let us define the set $T := UVU^\dagger$. Using Fact~\ref{fact:clifford_action_on_group}, we have that $T = UVU^\dagger = \calP^k \otimes \la Z_{k+1}, Z_{k+2}, \ldots, Z_{k+m}\rangle$. It can be checked that for any $\sigma \in V$ satisfying $|\la \psi | \sigma | \psi \ra|^2 \geq \gamma/4$ has a corresponding Pauli $\sigma_t \in T$ such that $|\la \psi' | \sigma_t | \psi' \ra|^2 \geq \gamma/4$ where $\ket{\psi'} = U \ket{\psi}$, as follows
\begin{align}
    \frac{\gamma}{2} \leq |\la \psi | \sigma | \psi \ra|^2  =  |\la \psi |U^\dagger U \sigma U^\dagger U | \psi \ra|^2 =|\la \psi' | U \sigma U^\dagger | \psi' \ra|^2 = |\la \psi' | \sigma_t | \psi' \ra|^2.
\end{align}

Using Fact~\ref{fact:size_set_anticommuting_paulis}, we can say there exists $P_1,\ldots, P_{2k+1} \in T (\text{or } V)$ such that $P_a P_b + P_b P_a = 0$ for all $a\neq b$ and $a,b \in [2k+1]$. We will now comment on the value of $k$ by commenting on the size of the maximal set of anti-commuting Paulis in the set $S$ (defined above in Theorem~\ref{thm:stabilizer_covering_group}). For this purpose, we will require the following folklore uncertainty relation. 
\begin{lemma}\label{lemma:uncertainty_relation}
For every $\rho$, let $P_1, P_2, \ldots P_M$ be a set of anti-commuting Paulis. Then $\sum_{j = 1}^{M} \Tr(\rho P_j)^2 \leq~1.$
\end{lemma}
\begin{proof}
For $y \in \R^{M}$, define the Hamiltonian $H(y) = \sum_{a=1}^{M} y_a P_a$. Then
\begin{equation*}
    H(y)^2 = \frac{1}{2}\sum_{a,b = 1}^{M} y_a y_b (P_a P_b + P_b P_a) = \norm{y}_2^2 I,
\end{equation*}
where we have used that the Paulis anti-commute in the second equality. Thus, $\norm{H(y)}_2 = \norm{y}_2$. Denote $w_a = \la \psi | P_a | \psi \ra$ and $w = (w_1,\ldots,w_K)$. We observe
\begin{equation*}
    \norm{w}_2^2 = \sum_{a=1}^K w_a^2 = \Tr(\rho H(w)) \leq \norm{H(w)}_2 = \norm{w}_2.
\end{equation*}
We thus obtain $\norm{w}_2 \leq 1$ which concludes the proof.
\end{proof}

\begin{proof}[Proof of Theorem~\ref{thm:stabilizer_covering_group}.] We have previously noted that $T \equiv \calP^k \times \la Z_{k+1},\ldots,Z_{k+m} \ra$. 
% Using Claim~\ref{claim:nac_V}, we have that $k \leq O(\poly(1/\gamma))$. 
Let us write the $4^k$ Paulis in $\calP^k = \{\tau_1,\tau_2,\ldots,\tau_{4^k}\}$. We choose $\tilde{S}_a = \la \tau_a \otimes I^{n-k} \ra \times \la Z_{k+1}, \ldots Z_{k+m} \ra \subseteq \calP^n$ for each $a \in [4^k]$. Moreover, we can extend $\tilde{S}_a$ to a stabilizer group $S_a$ on $n$-qubits. We then have 
\begin{equation}
    V \subseteq \bigcup \limits_{a=1}^{M} U^\dagger S_a U,
\end{equation}
where $M = 4^k$. Since $k=\nac(2S')\cdot \poly(1/\gamma)$, we have that  $M=\exp\big(\nac(2S')\cdot \poly(1/\gamma)\big)$. Moreover, the action of a Clifford circuit $U$ on a stabilizer group $S_a$ is to take it to another stabilizer group.
\end{proof}
%Moreover, $V \cap S \subseteq \calP^k \otimes \la Z_{k+1},\ldots,Z_{k+m} \ra$.

\paragraph{Conjecture implies covering.}  At this point  Theorem~\ref{app_thm:stabilizer_covering_V_combinatorial} implies that the stabilizer covering of $V$ is at most $O(\exp(\nac(2S') \cdot \gamma^{-C})$ for some constant $C>1$. The conjecture furthermore implies that $\nac(2S')\leq \poly(1/\gamma)\nac(S')$ and $\nac(S')\leq 1/\gamma$ follows from the uncertainty relation~(Lemma~\ref{lemma:uncertainty_relation}), so the overall stabilizer covering of $V$ has size at most $\exp(\poly(1/\gamma))$. We remark that the difference between this proof here and that of Theorem~\ref{thm:stabilizer_covering_group} that we presented in Section~\ref{sec:item3} is that the proof presented in the appendix is purely combinatorial in nature whereas the proof of Theorem~\ref{thm:stabilizer_covering_group} requires 
 analyzing the contributions across expectation values of different Weyl operators to $\Exp_{x \in V}[2^n p_\Psi(x)]$. 

\section*{B Alternate proof for item $3$}\label{app_sec:alt_proof_item3}
In this section, we give an alternate proof of Fact~\ref{eq:ub_sum1} that was crucially required to prove Theorem~\ref{thm:stabilizer_covering_group} which comments on the stabilizer covering of the subgroup $V$. This proof involves routine calculations starting from first principles and is fairly simple.

Let us recall some notation defined in Section~\ref{sec:item3}. Using Using Fact~\ref{fact:clifford_action_on_group}, we showed there exists a Clifford unitary $U$ such that we can define a set $T = UVU^\dagger = \calP^k \times \la Z_{k+1}, Z_{k+2}, \ldots, Z_{k+m}\rangle$. Moreover, we had
\begin{align*}
\frac{1}{|V|} \sum_{W_x \in \calP^k \times \la Z_{k+1},\ldots,Z_{k+m}\ra} 2^n p_{\tilde{\Psi}}(x) \geq \poly(\gamma)
\end{align*}
where the state $\tilde{\Psi}$ is the state obtained after application of the Clifford unitary $U$ i.e., $\tilde{\Psi} = U \Psi U^\dagger$. Let us denote $\tilde{\Psi}_A$ as the reduced density matrix (RDM) of $\tilde{\Psi}$ acting on system $A$ defined over the first $k$ qubits and $\tilde{\Psi}_B$ as the RDM of $\tilde{\Psi}$ acting on subsystem $B$ defined over the last $(n-k)$ qubits. 

We are now ready to prove Fact~\ref{eq:ub_sum1}, restated here for convenience.
\begin{fact}
\label{app:proofoffact}
We have the following
$$
\sum_{W_y \in \calP^k, W_z \in \calP_Z^m} \Tr\left( (W_y \otimes W_z) \tilde{\Psi}\right)^2  \leq 2^{m+k} \,.
$$
\end{fact}
\begin{proof}
In the computational basis, we have $\ket{\tilde{\psi}}=\sum_{a \in \mathbb{F}_2^k,b \in \mathbb{F}_2^{n-k}}s_{a,b}\ket{a,b}$ with $s_{a,b}$ denoting the amplitudes. The RDM of $\tilde{\Psi}$ over the first $k$-qubits or subsystem $A$ is~then
\begin{align*}
    \widetilde{\Psi}_A&=\sum_c \Big(\id \otimes  \bra{c}\Big)\widetilde{\Psi}\Big(\id \otimes \ket{c}\Big)
    =\sum_c\sum_{a,a',b,b'}s^*_{a,b}s_{a',b'} \big(\id \otimes  \bra{c}\big)\ketbra{a,b}{a',b'}\big(\id \otimes \ket{c}\big)
    =\sum_{a,a',b}s^*_{a,b}s_{a',b} \ketbra{a}{a'}.
\end{align*}
Let us define $   \widetilde{\Psi}_{A,b}=\sum_{a,a'}s^*_{a,b}s_{a',b} \ketbra{a}{a'}.
$ We can then show
\begin{align}
&\sum_{W_y \in \calP^k, W_z \in \calP_Z^m} \Tr\left( (W_y \otimes W_z) \tilde{\Psi}\right)^2\\
&=\sum_{W_y \in \calP^k, W_z \in \calP_Z^m} \Big|\sum_{a,b,a',b'}s^*_{a,b}s_{a',b'}\langle a,b| (W_y \otimes W_z)|a',b'\rangle\Big|^2\\
&=\sum_{W_y \in \calP^k, W_z \in \calP_Z^m} \Big|\sum_{a,b,a'}s^*_{a,b}s_{a',b}(-1)^{z\cdot b}\langle a| W_y|a'\rangle\Big|^2\\
&=\sum_{W_y \in \calP^k, W_z \in \calP_Z^m} \sum_{a,b,a',a_1,b_1,a'_1}s^*_{a,b}s_{a',b}\overline{s^*_{a_1,b_1}s_{a'_1,b_1}}(-1)^{z\cdot (b+b_1)}\langle a| W_y|a'\rangle\overline{\langle a_1| {W_y}|a'_1\rangle}\\
&=2^m \sum_b\sum_{W_y \in \calP^k} \Big(\sum_{a,a'}s^*_{a,b}s_{a',b}\langle a| W_y|a'\rangle\Big)\Big(\sum_{a_1,a'_1}\overline{s^*_{a_1,b}s_{a'_1,b}\langle a_1| W_y|a'_1\rangle}\Big)\\
&=2^m\sum_b\sum_{W_y \in \calP^k} \Big|\sum_{a,a'}s^*_{a,b}s_{a',b}\langle a| W_y|a'\rangle\Big|^2\\
&=2^m \sum_b \sum_{W_y \in \calP^k} |\Tr(W_y\tilde{\Psi}_{A,b})|^2
\end{align}
where the fourth equality used $\sum_{W_z \in \calP_Z^m} (-1)^{z\cdot(b+b_1)} = 2^m [b = b_1]$.
We then observe that
\begin{align}
\Tr(\tilde{\Psi}_{A,b}\overline{\tilde{\Psi}_{A,b}}) &= \frac{1}{2^{2k}} \Tr \left[ \sum_{x,y \in \mathbb{F}_2^{2k}} \Tr(W_x \tilde{\Psi}_{A,b}) \overline{\Tr(W_y \tilde{\Psi}_{A,b})} W_x \overline{W_y} \right] \\
&= \frac{1}{2^{2k}} \sum_{x,y \in \mathbb{F}_2^{2k}} \Tr(W_x \tilde{\Psi}_{A,b}) \overline{\Tr(W_y \tilde{\Psi}_{A,b})} \Tr(W_x \overline{W_y}) \\
&= \frac{1}{2^{2k}} \sum_{y \in \mathbb{F}_2^{2k}} |\Tr(W_y \tilde{\Psi}_{A,b})|^2 2^k 
= \frac{1}{2^{k}} \sum_{y \in \mathbb{F}_2^{2k}} |\Tr(W_y \tilde{\Psi}_{A,b})|^2,
\end{align}
where the first equality used the Pauli expansion of  $\Phi$ as\footnote{Even though $\tilde{\Psi}_{A,b}$ may not be a valid quantum state, we can write the Pauli decomposition of any operator.} ${\Phi} = \frac{1}{2^k} \sum_{x \in \mathbb{F}_2^{2k}} \Tr(W_x {\Phi}) W_x$. Hence,
\begin{align}
\label{eq:equalitywzwy}
    \sum_{W_y \in \calP^k, W_z \in \la Z_{k+1},\ldots,Z_{k+m}\ra} \Tr\left( (W_y \otimes W_z) \tilde{\Psi}\right)^2= 2^{k+m}\sum_b  \Tr( \tilde{\Psi}_{A,b}\overline{\tilde{\Psi}_{A,b}}).
\end{align}
We know that
\begin{align}
\Tr(\widetilde{\Psi}_{A,b}\overline{\tilde{\Psi}_{A,b}})=\Tr(\sum_{a,a',c,c'}s^*_{a,b}s_{a',b}s_{c,b}s^*_{c',b} \ketbra{a}{a'}\ketbra{c}{c'})=\sum_{a,a'}s^*_{a,b}s_{a',b}s_{a',b}s^*_{a,b}  =\Big|\sum_{a}s_{a,b}^2\Big|^2
\end{align}
So we can upper bound
$$
\sum_b  \Tr( \tilde{\Psi}_{A,b}\overline{\tilde{\Psi}_{A,b}})=\sum_{b}\big|\sum_a s_{a,b}^2\big|^2\leq \sum_{b}\sum_a |s_{a,b}^2|^2\leq \sum_{a,b}|s_{a,b}|^2=1,
$$
where we used that $\sum_i a_i^2\leq (\sum_i  a_i)^2$ for all $a_i\geq 0$ in the inequality and the equality uses that $\ket{\tilde{\psi}}=\sum_{a,b}s_{a,b}\ket{a,b}$ is a valid quantum state so $\sum_{a,b}|s_{a,b}|^2=1$. Putting together Eq.~\eqref{eq:equalitywzwy} with the inequality above gives the fact statement.
\end{proof}
\end{document}